\documentclass[10pt]{article}

\usepackage{amssymb,mathtools,enumitem}
\usepackage{mathrsfs} 
\usepackage{overpic,subcaption}
\usepackage{hyperref}
\usepackage{todonotes}\reversemarginpar
\newcounter{CONT}

\usepackage{float}

\usepackage[top=3cm,bottom=4cm,left=3cm,right=3cm]{geometry}

\usepackage[utf8]{inputenc} 
\usepackage[T1]{fontenc}

\newtheorem{theorem}{Theorem}
\newtheorem{definition}{Definition}
\newtheorem{remark}{Remark}
\newtheorem{proposition}{Proposition}
\newtheorem{lemma}{Lemma}
\newtheorem{corollary}{Corollary}

\newenvironment{proof}{\noindent{\bf Proof.}}
{\hspace*{\fill}$\Box$\par\vspace{4mm}}

\newcommand{\qed} {}

\newcommand{\claimbegin}[1] {\par\noindent\underline{Proof of #1:}}
\newcommand{\claimend}[1] {\hfill\underline{End proof of #1.}\par\vspace*{2mm}}

\newcommand{\Left} {\text{Left}}
\newcommand{\Right} {\text{Right}}
\newcommand{\dist}{\textrm{dist}}
\newcommand{\cc} {w}
\newcommand{\MF} {\textit{MF}}

\begin{document}
\pagestyle{plain}

\title{How Vulnerable is an Undirected Planar Graph \\with respect to Max Flow}


\author{Lorenzo Balzotti\footnote{Dipartimento di Scienze Statistiche, Sapienza Universit\`a di Roma, p.le Aldo Moro 5, 00185 Roma, Italy. \texttt{lorenzo.balzotti@uniroma1.it}.}
\and{Paolo G. Franciosa\footnote{Dipartimento di Scienze Statistiche, Sapienza Universit\`a di Roma, p.le Aldo Moro 5, 00185 Roma, Italy. \texttt{paolo.franciosa@uniroma1.it}}}}

\date{}
\maketitle
%
%

\textbf{Abstract}
We study the problem of computing the vitality of edges and vertices with respect to the $st$-max flow in undirected planar graphs, where the vitality of an edge/vertex is the $st$-max flow decrease when the edge/vertex is removed from the graph. This allows us to establish the vulnerability of the graph with respect to the $st$-max flow.

We give efficient algorithms to compute an additive guaranteed approximation of the vitality of edges and vertices in planar undirected graphs. We show that in the general case high vitality values are well approximated in time close to the time currently required to compute $st$-max flow $O(n\log \log n)$. We also give improved, and sometimes optimal, results in the case of integer capacities. All our algorithms work in $O(n)$ space.

\vspace{2mm}


\section{Introduction}

Max flow problems have been intensively studied in the last 60 years, we refer to~\cite{ahuja-magnanti_1,ahuja-magnanti_2} for a comprehensive bibliography. Currently, the best known algorithms for general graphs~\cite{king-rao,orlin} compute the max flow between two vertices in $O(mn)$ time, where $m$ is the number of edges and $n$ is the number of vertices.

Italiano \emph{et al.}~\cite{italiano} solved the max flow problem for undirected planar graphs in $O(n\log\log n)$  time. For directed $st$-planar graphs (i.e., graphs allowing a planar embedding with  $s$ and  $t$ on the same face) finding a max flow was reduced by Hassin~\cite{hassin} to the single source shortest path (SSSP) problem, that can be solved in $O(n)$  time by the algorithm in~\cite{henzinger}.  For the planar directed case, Borradaile and Klein~\cite{borradaile-klein} presented an $O(n\log n)$  time algorithm.  In the special case of directed planar unweighted graphs, a linear time algorithm was proposed by Eisenstat and Klein~\cite{eisenstat-klein}.

The effect of edges deletion on the max flow value has been studied since 1963, only a few years after the seminal paper by Ford and Fulkerson~\cite{ford-fulkerson_1} in 1956. Wollmer~\cite{wollmer} presented a method for determining the most vital edge (i.e., the edge whose deletion causes the largest decrease of the max flow value) in a railway network. A more general problem was studied in~\cite{ratliff-sicilia}, where an enumerative approach is proposed for finding the $k$ edges whose simultaneous removal causes the largest decrease
in max flow. Wood~\cite{wood} showed that this problem is NP-hard in the strong sense, while its approximability has been studied in~\cite{alter-ergun,phillis}. 

In this paper we deal with the computation of  \emph{vitality} of edges and vertices with respect to the $st$-max flow, where $s$ and $t$ are two fixed vertices. The vitality of an edge $e$ (resp., of a vertex $v$) measures the $st$-max flow decrease observed after the removal of edge $e$ (resp., all edges incident to $v$) from the graph. 

A reasonable measure of the overall vulnerability of a network can be the number of edges/vertices with high vitality. So, if all edges and vertices have small vitality, then the graph is robust. We stress that verifying the robustness/vulnerability of the graph by using previous algorithms requires to compute the exact vitality of all edges and/or vertices. We refer to~\cite{alderson-brown,mattsson-jenelius,murray} for surveys on several kind of robustness and vulnerability problems discussed by an algorithmic point of view. 

A survey on vitality with respect to the max flow problems can be found in~\cite{ausiello-franciosa_1}. In the same paper, it is shown that for $st$-planar graphs (both directed or undirected) the vitality of all edges and all vertices can be found in optimal $O(n)$ time. Ausiello \emph{et al.}~\cite{ausiello-franciosa_2} proposed a recursive algorithm that computes the vitality of all edges in an undirected unweighted planar graph in $O(n \log n)$  time.

The vitality has not been studied directly in~\cite{italiano,lacki-sankowski}, but their dynamic algorithm leads to the following result.

\begin{theorem}[\cite{italiano,lacki-sankowski}]\label{th:italiano_dynamic}
Let $G$ be a planar graph with positive edge capacities. Then it is possible to compute the vitality of $h$ single edges or the vitality of a set of $h$ edges in $O(\min\{\frac{hn}{\log n}+n\log\log n,$ $hn^{2/3}\log^{8/3} n+n\log n\})$.
\end{theorem}

\paragraph{Our contribution} 
We propose fast algorithms for computing an additive guaranteed approximation of the vitality of all edges and vertices whose capacity is less than an arbitrary threshold $c$. Later, we explain that these results can be used to obtain a useful approximation of vitality for general distribution of capacities and in the case of power-law distribution. We stress that in real world applications we are usually interested in finding edges and or vertices with high vitality, i.e., edges or vertices whose removal involves relevant decrease on the max flow value.

Our main results are summarized in the following two theorems. For a graph $G$, we denote by $E(G)$ and $V(G)$ its set of edges and vertices, respectively. Let $c:E(G)\to\mathbb{R}^+$ be the \emph{edge capacity function}, we define the capacity $c(v)$ of a vertex $v$ as the sum of the capacities of all edges incident on $v$. Moreover, for every $x\in E(G)\cup V(G)$ we denote by $vit(x)$ its vitality with respect to the $st$-max flow. We show that we can compute a value $vit^\delta(x)$ in $(vit(x)-\delta,vit(x)]$ for any $\delta>0$.

\begin{theorem}\label{th:main_real_edge}
Let $G$ be a planar graph with positive edge capacities. Then for any $c,\delta>0$, we compute a value $vit^\delta(e)\in(vit(e)-\delta,vit(e)]$ for all $e\in E(G)$ satisfying $c(e)\leq c$, in $O(\frac{c}{\delta}n+n\log\log n)$  time.
\end{theorem}

\begin{theorem}\label{th:main_real_vertex}
Let $G$ be a planar graph with positive edge capacities. Then for any $c,\delta>0$, we compute a value $vit^\delta(v)\in(vit(v)-\delta,vit(v)]$ for all $v\in V(G)$ satisfying $c(v)\leq c$, in $O(\frac{c}{\delta}n+n\log n)$  time.
\end{theorem}

All our algorithms work in $O(n)$ space. To explain  the result stated in Theorem~\ref{th:main_real_edge}, we note that in the general case capacities are not bounded by any function of $n$. Thus, in order to obtain useful approximations, $c/\delta$ would seem to be very high. Despite this in many cases we can assume $c/\delta$ constant, implying that the time complexity of Theorem~\ref{th:main_real_edge} is equal to the best current time bound for computing the $st$-max flow. The following remark is crucial, where $c_{max}=\max_{e\in E(G)}c(e)$.

\begin{remark}\label{remark:bounding}[Bounding capacities]
We can bound all edge capacities higher than $\MF$ to $\MF$, obtaining a new bounded edge capacity function. This change has no impact on the $st$-max flow value or the vitality of any edge/vertex. Thus w.l.o.g., we can assume that $c_{max}\leq \MF$.
\end{remark}

By using Remark~\ref{remark:bounding} we can explain why $c/\delta$ can be assumed constant, we study separately the case of general distribution of capacities and the case of power-law distribution.

\noindent
$\bullet$ \emph{General distribution} (after bounding capacities as in Remark~\ref{remark:bounding}). If we set $c=c_{max}$ and $\delta=c/k$, for some constant $k$, then we obtain the capacities with an additive error less than $\MF/k$, because of Remark~\ref{remark:bounding}. In many applications this error is acceptable even for small values of $k$, e.g., $k=10,50,100$. In this way we obtain small percentage error of vitality for edges with high vitality---edges whose vitality is comparable with $\MF$---while edges with small vitality---edges whose vitality is smaller than $\MF/k$---are badly approximated. We stress that we are usually interested in high capacity edges, and that with these choices the time complexity is $O(n\log\log n)$, that is the time currently required for the computation of the $st$-max flow.

\noindent
$\bullet$ \emph{Power-law distribution} (after bounding capacities as in Remark~\ref{remark:bounding}). The previous method cannot be applied to power-law distribution because the most of the edges has capacity lower than $\MF/k$, even for high value of $k$. Thus we have to separate edges with high capacity and edges with low capacity. Let $c=\frac{c_{max}}{\ell}$ for some constant $\ell$ and let $H_c=\{e\in E(G)\,|\, c(e)>c\}$. By power-law distribution, $|H_c|$ is small even for high values of $\ell$, and thus we compute the exact vitality of edges in $H_c$ by Theorem~\ref{th:italiano_dynamic}. For edges with capacity less than $c$, we set $\delta=c/k$, for some constant $k$. By Remark~\ref{remark:bounding} we compute the vitality of these edges with an additive error less than $\frac{\MF}{k\ell}$. Again, the overall time complexity is equal or close to the time currently required for the computation of the $st$-max flow.

It is worth studying the case in which, regardless of the distribution of capacities, all edges have vitality less than $\MF/k$ even for an high value of $k$. By choosing $c=\MF$ and $\delta=\MF/k$, we may establish by using Theorem~\ref{th:main_real_edge} that all edges have vitality in $[0,\MF/k]$; this happens when $vit^\delta(e)=0$ for all $e\in E(G)$. We stress that this is not a useless result; indeed, we certify that all edges in the network have low vitality, so the network is robust.

To apply the same arguments to vertex vitality we need some observations. If $G$'s vertices have maximum degree $d$, then, after bounding capacities as in Remark~\ref{remark:bounding}, it holds $\max_{v\in V(G)}c(v)\leq d \MF$. Otherwise, we note that a real-world planar graph is expected to have few vertices with high degree (it is also implied by Euler formula). The exact vitality of these vertices can be computed by Theorem~\ref{th:italiano_dynamic} or by the following result.

\begin{theorem}\label{th:brutto}
Let $G$ be a planar graph with positive edge capacities. Then for any $S\subseteq V(G)$, we compute $vit(v)$ for all $v\in S$ in $O(|S|n+n\log\log n)$  time. 
\end{theorem}

If we denote by $E_S=\sum_{v\in S}deg(v)$, where $deg(v)$ is the degree of vertex $v$, then the result in Theorem~\ref{th:brutto} is more efficient than the result given in Theorem~\ref{th:italiano_dynamic} if either $|S|<\log n$ and $E_S>|S|\log n$ or $|S|\geq\log n$ and $E_S>\frac{|S|n^{1/3}}{\log^{8/3}}$.

\paragraph{Small integer case} 
In the case of integer capacity values that do not exceed a small constant, or in the more general case in which   capacity values are integers with bounded sum we also prove the following corollaries.

\begin{corollary}\label{cor:integer}
Let $G$ be a planar graph with integer edge capacity and let $L$ be the sum of all the edges capacities. Then
\begin{itemize}\itemsep0em
\item for any $H\subseteq E(G)\cup V(G)$, we compute $vit(x)$ for all $x\in H$, in $O(|H|n+L)$ time,
\item for any $c\in\mathbb{N}$, we compute $vit(e)$ for all $e\in E(G)$ satisfying $c(e)\leq c$, in $O(cn+L)$ time.
\end{itemize}
\end{corollary}

\begin{corollary}\label{cor:unweighted}
Let $G$ be a planar graph with unit edge capacity. Let $n_{>d}$ be the number of vertices whose degree is greater than $d$. We can compute the vitality of all edges in $O(n)$ time and the vitality of all vertices in $O(\min\{n^{3/2},n(n_{>d}+d+\log n)\})$ time.
\end{corollary}

\begin{corollary}\label{cor:bounded}
Let $G$ be a planar graph with unit edge capacity where only a constant number of vertices have degree greater than a fixed constant $d$. Then we can compute the  vitality of all vertices in $O(n)$ time.
\end{corollary}


\paragraph{Our approach} We adopt Itai and Shiloach's approach~\cite{itai-shiloach}, that first computes a modified version $D$ of a dual graph of $G$, then reduces the computation of the max flow to the computation of  non-crossing shortest paths between pairs of vertices of the infinite face of $D$. We first study the effect on $D$ of an edge or a vertex removal in  $G$, showing that computing the vitality of an edge or a vertex can be reduced to computing some distances in $D$ (see Proposition~\ref{prop:vitality_edges} and Proposition~\ref{prop:vitality_vertices}).

Then we determine required distances by solving SSSP instances. To decrease the cost we use a divide and conquer strategy: we slice $D$ in regions delimited by some of the non-crossing shortest paths computed above. We choose non-crossing shortest paths with similar lengths, so that we compute an additive guaranteed approximation of each distance by looking into a single region instead of examining the whole graph $D$ (see Lemma~\ref{lemma:Omega_i}).

Finally we have all the machinery to compute an approximation of required distances of Proposition~\ref{prop:vitality_edges} and Proposition~\ref{prop:vitality_vertices} and obtain edge and vertex vitalities. 

\paragraph{Structure of the paper} In Section~\ref{sec:max_flow_planar} we report main results about how to compute max flow in planar graphs; we focus on the approach in~\cite{itai-shiloach} on which our algorithms are based. In Section~\ref{sec:main_results} we show our theoretical results that allow us to compute edge and vertex vitality. In Section~\ref{sec:divide_and_conquer} we explain our divide and conquer strategy. In Section~\ref{sec:complexity_edge} we state our main result about edge vitality. Vertex vitality is described in Section~\ref{sec:complexity_vertex} and in Section~\ref{sec:small_integer} we obtain some corollaries about planar graphs with small integer capacities. Finally, in Section~\ref{sec:conclusions} conclusions and open problems are given.

\section{Max flow in planar graphs}\label{sec:max_flow_planar}

In this section we report some well-known results concerning max flow, focusing on planar graphs.

Given a connected undirected graph $G=(V(G),E(G))$ with $n$ vertices, we denote an edge $e=\{i,j\}\in E(G)$ by the shorthand notation $ij$, and we define $\dist_G(u,v)$ as the length of a shortest path in $G$ joining vertices $u$ and $v$. Moreover, for two sets of vertices $S,T\subseteq V(G)$, we define $\dist_G(S,T)=\min_{u\in S,v\in T}\dist_G(u,v)$. We write for short $v\in G$ and $e\in G$ in place of $v\in V(G)$ and $e\in E(G)$, respectively. We say that a path $p$ is an $ab$ path if its extremal vertices are $a$ and $b$.

Let $s,t\in G$, $s\neq t$, be two fixed vertices. A \emph{feasible flow} in $G$ assigns to each edge $e=ij\in G$ two real values $x_{ij}\in[0, c(e)]$ and $x_{ji}\in [0, c(e)]$ such that: 
$\sum_{j:ij\in E(G)}x_{ij}=\sum_{j:ij\in E(G)}x_{ji}$, for each $i\in V(G)\setminus\{s,t\}$.
The \emph{flow from s to t} under a feasible flow assignment $x$ is defined as
$F(x)=\sum_{j:sj\in E(G)}x_{sj}-\sum_{j:sj\in E(G)}x_{js}$. 
The maximum flow from $s$ to $t$, denoted by $\MF$, is the maximum value of $F(x)$ over all feasible flow assignments $x$.

An \emph{st-cut} is a partition of $V(G)$ into two subsets $S$ and $T$ such that $s\in S$ and $t\in T$. The capacity of an $st$-cut is the sum of the capacities of the edges $ij\in E(G)$ such that $|S\cap\{i,j\}|=1$ and $|T\cap\{i,j\}|=1$. The well known Min-Cut Max-Flow theorem~\cite{ford-fulkerson_1} states that the maximum flow from $s$ to $t$ is equal to the capacity of a minimum $st$-cut for any weighted graph $G$.

We denote by $G-e$ the graph $G$ after the removal of edge $e$. Similarly, we denote by $G-v$ the graph $G$ after the removal of vertex $v$ and all edges adjacent to $v$.

\begin{definition}\label{definition:vitality}
The vitality $vit(e)$ (resp., $vit(v)$) of an edge $e$ (resp., vertex $v$) with respect to the maximum flow from $s$ to $t$, according to the general concept of vitality in~\cite{koschutzki-lehmann}, is defined as the difference between the maximum flow in $G$ and the maximum flow in $G-e$ (resp., $G-v$).
\end{definition}

We deal with planar undirected graphs. A \emph{plane graph} is a planar graph with a fixed embedding. The dual of a plane undirected graph $G$ is an undirected planar multigraph $G^*$ whose vertices correspond to faces of $G$ and such that for each edge $e$ in $G$ there is an edge $e^*=\{u^*,v^*\}$ in $G^*$, where $u^*$ and $v^*$ are the vertices in $G^*$ that correspond to faces $f$ and $g$ adjacent to $e$ in $G$. Length $w(e^*)$ of $e^*$ equals the capacity of $e$, moreover, for a subgraph $H$ of $G^*$ we define $w(H)=\sum_{e\in H}w(e)$. 

We fix a planar embedding of the graph, and we work on the dual graph $G^*$ defined by this embedding. A vertex $v$ in $G$ generates a face in $G^*$ denoted by $f^*_v$. We choose in $G^*$ a vertex $v^*_s$ in $f^*_s$ and a vertex $v^*_t$ in $f^*_t$. A cycle in the dual graph $G^*$ that separates vertex $v^*_s$ from vertex $v^*_t$ is called an \emph{$st$-separating cycle}. Moreover, we choose a shortest path $\pi$ in $G^*$ from $v^*_s$ to $v^*_t$.

\begin{proposition}[see~\cite{itai-shiloach,reif}]\label{prop:st-min-cut}
A (minimum) $st$-cut in $G$ corresponds to a (shortest) cycle in $G^*$ that separates vertex $v^*_s$ from vertex $v^*_t$.
\end{proposition}

\subsection{Itai and Shiloach's approach/decomposition}

According to the approach by Itai and Shiloach in~\cite{itai-shiloach} used to find a min-cut by searching for minimum $st$-separating cycles, graph $G^*$ is ``cut'' along the fixed shortest path $\pi$ from $v^*_s$ to $v^*_t$, obtaining graph $D_G$, in which each vertex $v^*_i$ in $\pi$ is split into two vertices $x_i$ and $y_i$; when no confusion arises we omit the subscript $G$. In Figure~\ref{fig:G} there is a plane graph $G$ in black continuous lines and in Figure~\ref{fig:pi_double_and_D} on the right graph $D$. Now we explain the construction of the latter.

Let us assume that $\pi=\{v_1^*,v_2^*,\ldots,v_k^*\}$, with $v^*_1=v^*_s$ and $v^*_k=v^*_t$. For convenience, let $\pi_x$ be the duplicate of $\pi$ in ${D}$ whose vertices are $\{x_1,\ldots,x_k\}$ and let $\pi_y$ be the duplicate of $\pi$ in ${D}$ whose vertices are $\{y_1,\ldots,y_k\}$. For any $i\in[k]$, where $[k]=\{1,\ldots,k\}$, edges in $G^*$ incident on each $v^*_i$ from below $\pi$ are moved to $y_i$ and edges incident on $v^*_i$ from above $\pi$ are moved to $x_i$. Edges incident on $v^*_s$ and $v^*_t$ are considered above or below $\pi$ on the basis of two dummy edges: the first joining $v^*_s$ to a dummy vertex $\alpha$ inside face $f^*_s$ and the second joining $v^*_t$ to a dummy vertex $\beta$ inside face $f^*_t$. In Figure~\ref{fig:G} there is a graph $G$ in black continuous line, $G^*$ in red dashed lines and shortest path $\pi$ from $v^*_1$ to $v^*_k$. In Figure~\ref{fig:pi_double_and_D}, on the left there are the graph $G$ and $G^*$ of Figure~\ref{fig:G} where path $\pi$ is doubled.

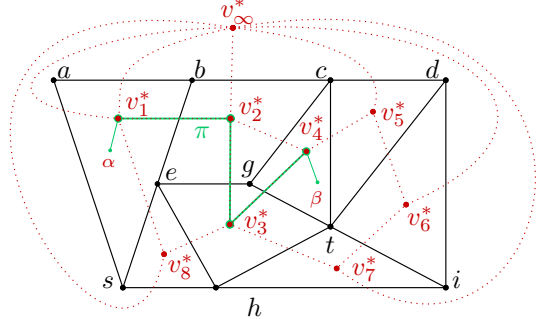
\begin{figure}[h]
\captionsetup[subfigure]{justification=centering}
\centering
\begin{subfigure}[t]{7cm}
\begin{overpic}[width=7cm,percent]{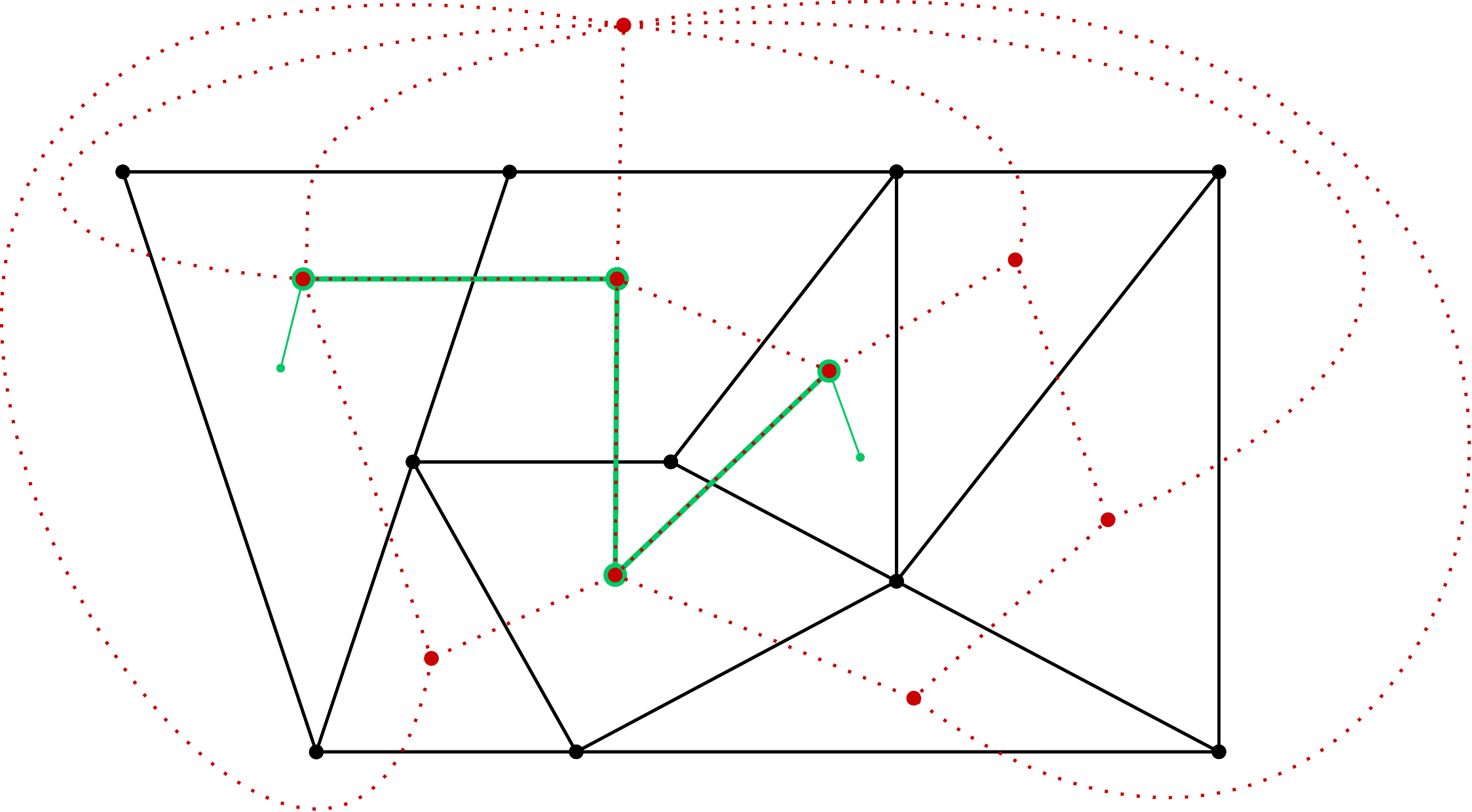}
\put(17.5,27){\color[rgb]{0.784,0,0}\scriptsize $\alpha$}
\put(57,20){\color[rgb]{0.784,0,0}\scriptsize $\beta$}

\put(22,38){\color[rgb]{0.784,0,0}$v_1^*$}
\put(43.5,37){\color[rgb]{0.784,0,0}$v_2^*$}
\put(44.5,15.5){\color[rgb]{0.784,0,0}$v_3^*$}
\put(55,33.3){\color[rgb]{0.784,0,0}$v_4^*$}
\put(70.5,36){\color[rgb]{0.784,0,0}$v_5^*$}
\put(75.5,16){\color[rgb]{0.784,0,0}$v_6^*$}
\put(65,6.5){\color[rgb]{0.784,0,0}$v_7^*$}
\put(30,7){\color[rgb]{0.784,0,0}$v_8^*$}
\put(40,55.5){\color[rgb]{0.784,0,0}$v_\infty^*$}

\put(9,44){$a$}
\put(35,44){$b$}
\put(58,44){$c$}
\put(79,44){$d$}
\put(29.5,24.5){$e$}
\put(44,26){$g$}
\put(45,-1.3){$h$}
\put(84,4){$i$}
\put(17.5,4){$s$}
\put(60,10.7){$t$}

\put(35,32.5){\color[rgb]{0,0.784,0.392}$\pi$}

\end{overpic}
\end{subfigure}	
\caption{graph $G$ in black continuous line, $G^*$ in red dashed lines, shortest path $\pi$ from $v^*_s$ ($v^*_1$) to $v^*_t$ ($v^*_k$) in green, $\alpha$ and $\beta$ are dummy vertices.}
\label{fig:G}
\end{figure}

For each $e^*\in\pi$, we denote by $e^*_x$ the copy of $e^*$ in $\pi_x$ and $e^*_y$ the copy of $e^*$ in $\pi_y$. Note that each $v\in V(G)\setminus\{s,t\}$ generates a face $f^D_v$ in ${D}$. There are not faces $f^D_s$ and $f^D_t$ because the dummy vertices $\alpha$ and $\beta$ are inside faces $f^*_s$ and $f^*_t$, respectively. Both faces $f^*_s$ and $f^*_t$ ``correspond'' in ${D}$ to the leftmost $x_1y_1$ path and to the rightmost $x_ky_k$ path, respectively. Since we are not interested in removing vertices $s$ and $t$, then faces $f^D_s$ and $f^D_t$ are not needed in $D$.  In Figure~\ref{fig:pi_double_and_D}, on the right there is graph $D$ built on $G$ in Figure~\ref{fig:G}.

If $e^*\not\in\pi$, then we denote the corresponding edge in $D$ by $e^D$. Similarly, if $v^*_i\not\in\pi$ (that is, $i>k$), then we denote the corresponding vertex in $D$ by $v^D_i$. 


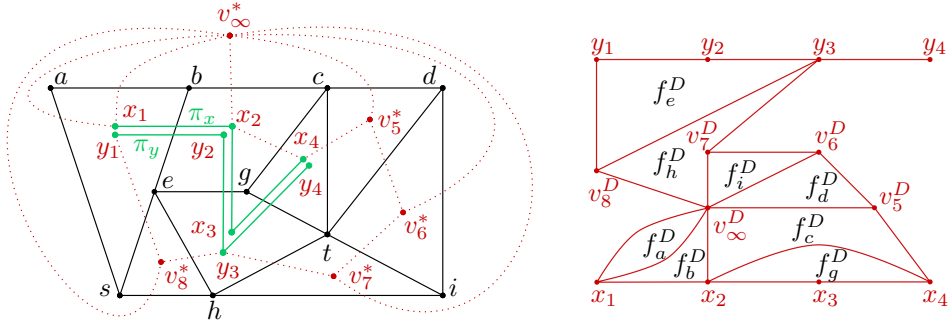
\begin{figure}[h]
\captionsetup[subfigure]{justification=centering}
\centering
\begin{subfigure}{7cm}
\begin{overpic}[width=7cm,percent]{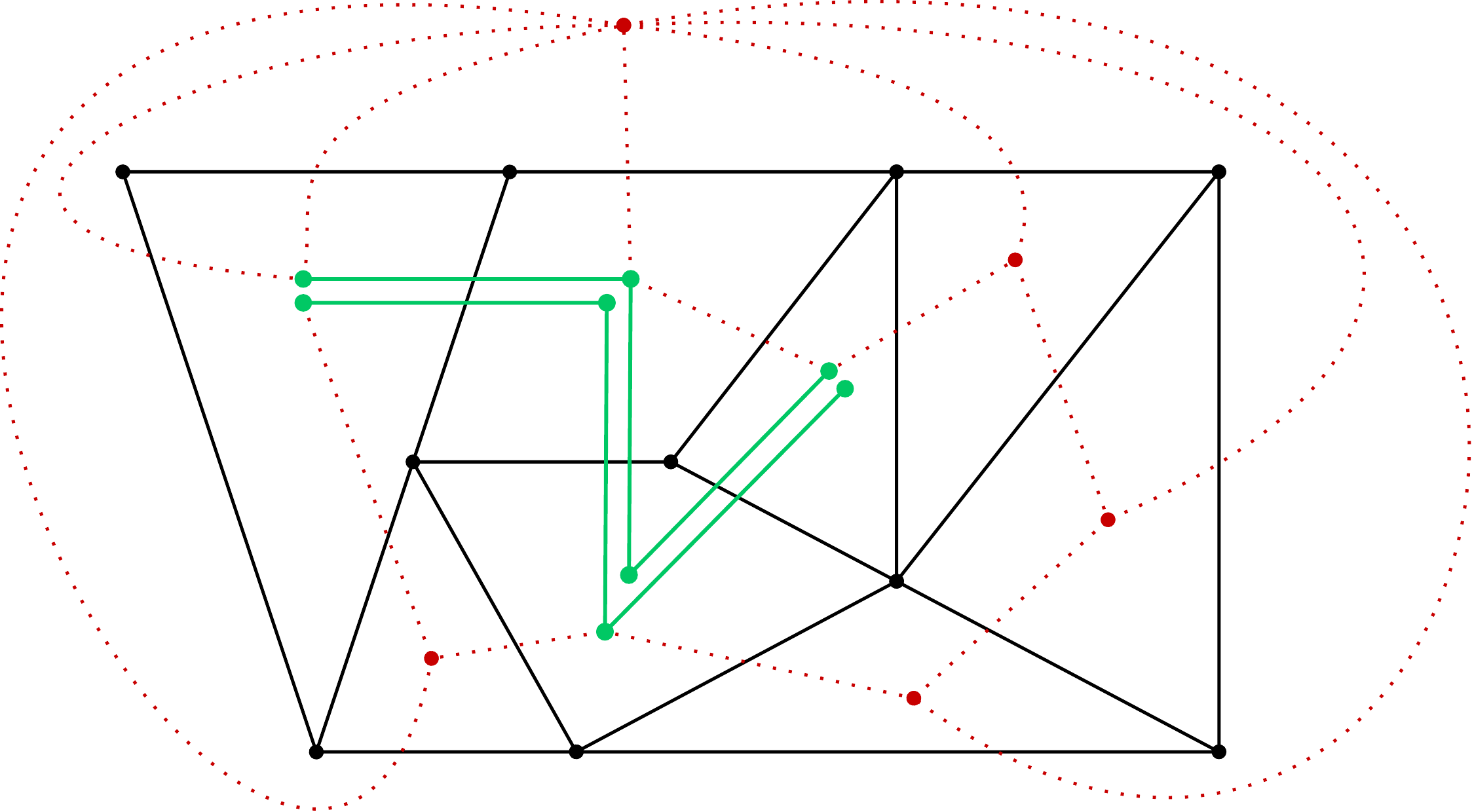}

\put(34.5,37.5){\color[rgb]{0,0.784,0.350}$\pi_x$}
\put(24,31.5){\color[rgb]{0,0.784,0.350}$\pi_y$}

\put(22,38){\color[rgb]{0.784,0,0}$x_1$}
\put(43.5,37){\color[rgb]{0.784,0,0}$x_2$}
\put(35,16){\color[rgb]{0.784,0,0}$x_3$}
\put(54.5,32.8){\color[rgb]{0.784,0,0}$x_4$}
\put(70.5,36){\color[rgb]{0.784,0,0}$v_5^*$}
\put(75.5,16){\color[rgb]{0.784,0,0}$v_6^*$}
\put(65,6.5){\color[rgb]{0.784,0,0}$v_7^*$}
\put(30,6.5){\color[rgb]{0.784,0,0}$v_8^*$}
\put(40,56.3){\color[rgb]{0.784,0,0}$v_\infty^*$}

\put(17,31){\color[rgb]{0.784,0,0}$y_1$}
\put(35,31){\color[rgb]{0.784,0,0}$y_2$}
\put(39.8,8.2){\color[rgb]{0.784,0,0}$y_3$}
\put(55.5,24){\color[rgb]{0.784,0,0}$y_4$}

\put(9,44.2){$a$}
\put(35,44.2){$b$}
\put(58,44.2){$c$}
\put(79,44.2){$d$}
\put(29.5,24.5){$e$}
\put(44,26){$g$}
\put(38,-0.5){$h$}
\put(84,4){$i$}
\put(17.5,4){$s$}
\put(60,10.7){$t$}

\end{overpic}
\end{subfigure}
\qquad
\begin{subfigure}{4.5cm}
\begin{overpic}[width=4.5cm,percent]{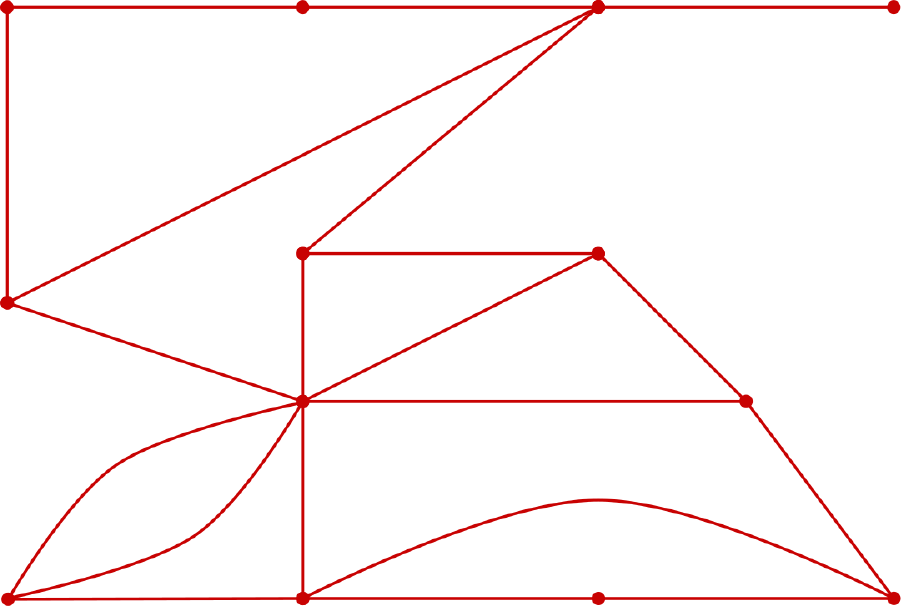}

\put(-1.2,69.6){\color[rgb]{0.784,0,0}$y_1$}
\put(31.3,69.6){\color[rgb]{0.784,0,0}$y_2$}
\put(64.3,69.6){\color[rgb]{0.784,0,0}$y_3$}
\put(96.8,69.6){\color[rgb]{0.784,0,0}$y_4$}

\put(84,24){\color[rgb]{0.784,0,0}$v^D_5$}
\put(65,42){\color[rgb]{0.784,0,0}$v^D_6$}
\put(27,42){\color[rgb]{0.784,0,0}$v^D_7$}
\put(-2,25){\color[rgb]{0.784,0,0}$v^D_8$}
\put(35,15.5){\color[rgb]{0.784,0,0}$v^D_\infty$}

\put(-1.2,-5){\color[rgb]{0.784,0,0}$x_1$}
\put(31.3,-5){\color[rgb]{0.784,0,0}$x_2$}
\put(64.3,-5){\color[rgb]{0.784,0,0}$x_3$}
\put(96.8,-5){\color[rgb]{0.784,0,0}$x_4$}

\put(14,9.9){${f}_a^D$}
\put(23,4){${f}_b^D$}
\put(58,14.5){${f}_c^D$}
\put(62,27){${f}_d^D$}
\put(17,55){${f}_e^D$}
\put(65,4){${f}_g^D$}
\put(17,34){${f}_h^D$}
\put(38,30.5){${f}_i^D$}
\end{overpic}
\end{subfigure}		
\caption{on the left green path $\pi$ is doubled into paths $\pi_x$ and $\pi_y$, and edges incident on $x_1,y_1,x_4,y_4$ in $G^*$ are moved according to the dummy vertices $\alpha$ and $\beta$ in Figure~\ref{fig:G}. On the right graph $D$.}
\label{fig:pi_double_and_D}
\end{figure}

\section{Our theoretical results}\label{sec:main_results}

In this section we show our main theoretical results (Proposition~\ref{prop:vitality_edges} and Proposition~\ref{prop:vitality_vertices}) that allow us to compute edge and vertex vitality. In Subsection~\ref{sub:effects_deleting_edge_or_vertex} we show the effects in $G^*$ and $D$ of removing an edge or a vertex from $G$. In Subsection~\ref{sub:single-crossing} we prove that we can focus only on $st$-separating cycles that cross $\pi$ exactly once, and in Subsection~\ref{sub:what_we_need} we state the two main propositions about edge and vertex vitality.

\subsection{Effects on $G^*$ and ${D}$ of deleting an edge or a vertex of $G$}\label{sub:effects_deleting_edge_or_vertex}


We observe that  removing an edge $e$ from $G$ corresponds to  contracting  endpoints of $e^*$ into one vertex in $G^*$. With respect to $D$, if $e^*\not\in\pi$, then the removal of $e$ corresponds to the contraction into one vertex of endpoints of $e^D$. If $e^*\in\pi$, then both copies of $e^*$ have to be contracted. In Figure~\ref{fig:G_setminus_edge} we show  the effects of removing edge $eg$ from graph $G$ in Figure~\ref{fig:G}.

Let $v$ be a vertex of $V(G)$. Removing $v$ corresponds to contracting vertices of face $f^*_v$ in $G^*$ into a single vertex. In ${D}$, if $f^*_v$ and $\pi$ have no common vertices, then all vertices of $f^D_v$ are contracted into one. Otherwise $f^*_v$ intersects $\pi$ on vertices $\bigcup_{i\in I}\{v^*_i\}$ for some non empty set $I\subseteq[k]$. Then all vertices of $f^D_v$ are contracted into one vertex, all vertices of $\bigcup_{i\in I}\{x_i\}$ not belonging to ${f}^D_v$ are contracted into another vertex and all vertices of $\bigcup_{i\in I}\{y_i\}$ not belonging to ${f}^D_v$ are contracted into a third vertex. For convenience, we define $q_{f^D_v}^x=(\bigcup_{i\in I_v}\{x_i\})\setminus V(f^D_v)$ and $q_{f^D_v}^y=(\bigcup_{i\in I_v}\{y_i\})\setminus V(f^D_v)$. To better understand these definitions, see Figure~\ref{fig:xy_face}.
In Figure~\ref{fig:G_setminus_vertex} it is shown what happens when we remove vertex $g$ of graph $G$ in Figure~\ref{fig:G}.

\begin{figure}[h]
\captionsetup[subfigure]{justification=centering}
\centering
\begin{subfigure}[t]{7cm}
\begin{overpic}[width=7cm,percent]{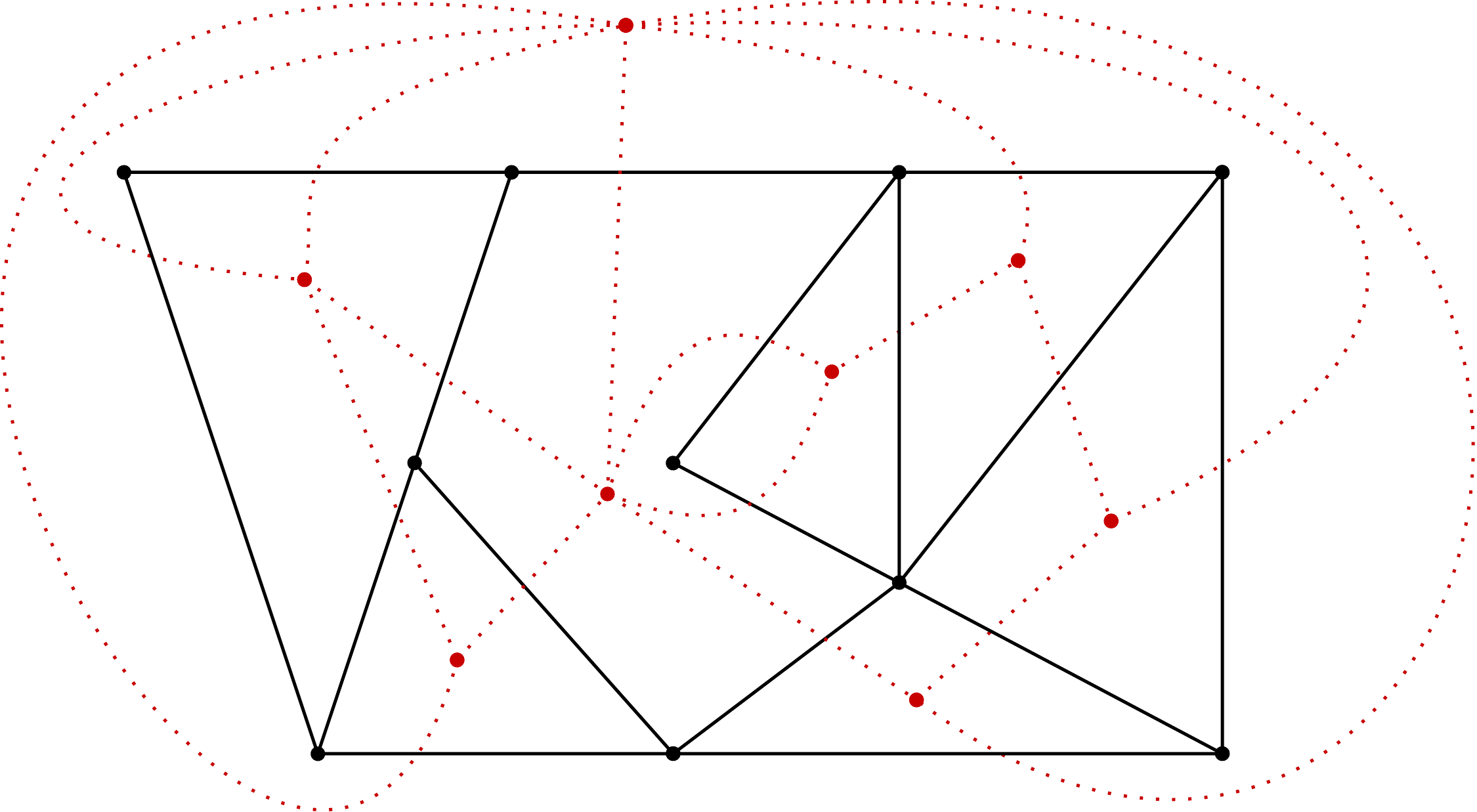}

\put(39,15){\color[rgb]{0.784,0,0}$v^*_{\text{new}}$}
\put(22,36){\color[rgb]{0.784,0,0}$v_1^*$}
\put(55,33.3){\color[rgb]{0.784,0,0}$v_4^*$}
\put(70.5,36){\color[rgb]{0.784,0,0}$v_5^*$}
\put(75.5,16){\color[rgb]{0.784,0,0}$v_6^*$}
\put(65,6.5){\color[rgb]{0.784,0,0}$v_7^*$}
\put(31.5,7){\color[rgb]{0.784,0,0}$v_8^*$}
\put(40,55.5){\color[rgb]{0.784,0,0}$v_\infty^*$}

\put(9,44){$a$}
\put(35,44){$b$}
\put(58,44){$c$}
\put(79,44){$d$}
\put(29.5,24.5){$e$}
\put(44,26){$g$}
\put(45,-1.3){$h$}
\put(84,4){$i$}
\put(17.5,4){$s$}
\put(60,10.7){$t$}

\end{overpic}
\end{subfigure}
\qquad
\begin{subfigure}[t]{4.5cm}
\begin{overpic}[width=4.5cm,percent]{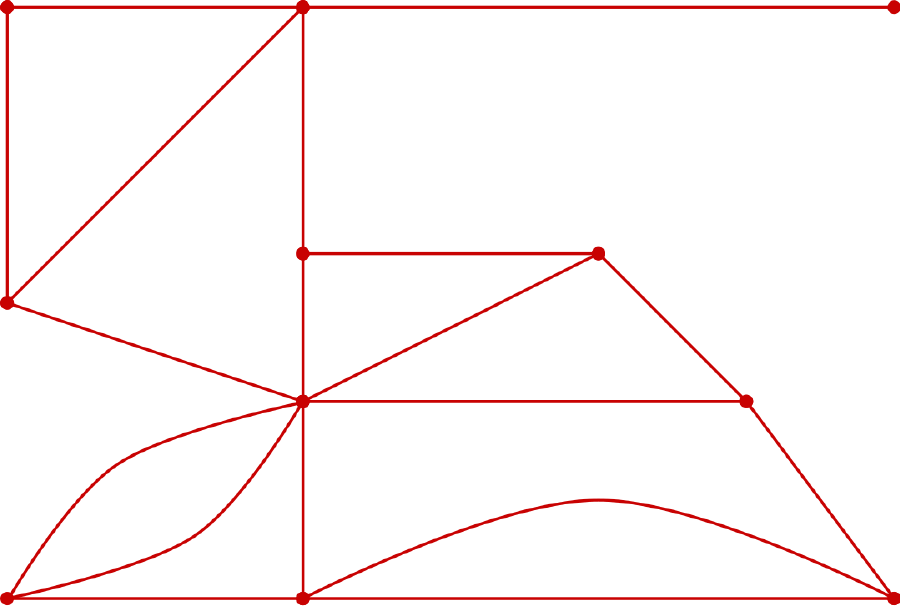}
\put(-1.2,69.6){\color[rgb]{0.784,0,0}$y_1$}
\put(31.3,69.6){\color[rgb]{0.784,0,0}$y_{\text{new}}$}

\put(96.8,69.6){\color[rgb]{0.784,0,0}$y_4$}

\put(-1.2,-5){\color[rgb]{0.784,0,0}$x_1$}
\put(31.3,-5){\color[rgb]{0.784,0,0}$x_{\text{new}}$}
\put(96.8,-5){\color[rgb]{0.784,0,0}$x_4$}

\put(84,24){\color[rgb]{0.784,0,0}$v^D_5$}
\put(65,42){\color[rgb]{0.784,0,0}$v^D_6$}
\put(35,42){\color[rgb]{0.784,0,0}$v^D_7$}
\put(-2,25){\color[rgb]{0.784,0,0}$v^D_8$}
\put(35,15.5){\color[rgb]{0.784,0,0}$v^D_\infty$}

\put(14,9.9){${f}_a^D$}
\put(23,4){${f}_b^D$}
\put(58,14.5){${f}_c^D$}
\put(62,27){${f}_d^D$}
\put(12,55){${f}_e^D$}
\put(65.5,4){${f}_g^D$}
\put(17,34){${f}_h^D$}
\put(38,30.5){${f}_i^D$}

\end{overpic}
\end{subfigure}		
\caption{starting from graph $G$ in Figure~\ref{fig:G}, we show on the left graph $G-eg$ and $(G-eg)^*$, and graph $D_{G-eg}$ on the right.}
\label{fig:G_setminus_edge}
\end{figure}

\begin{figure}[h]
\captionsetup[subfigure]{justification=centering}
\centering
	\begin{subfigure}{5cm}
\begin{overpic}[width=5cm,percent]{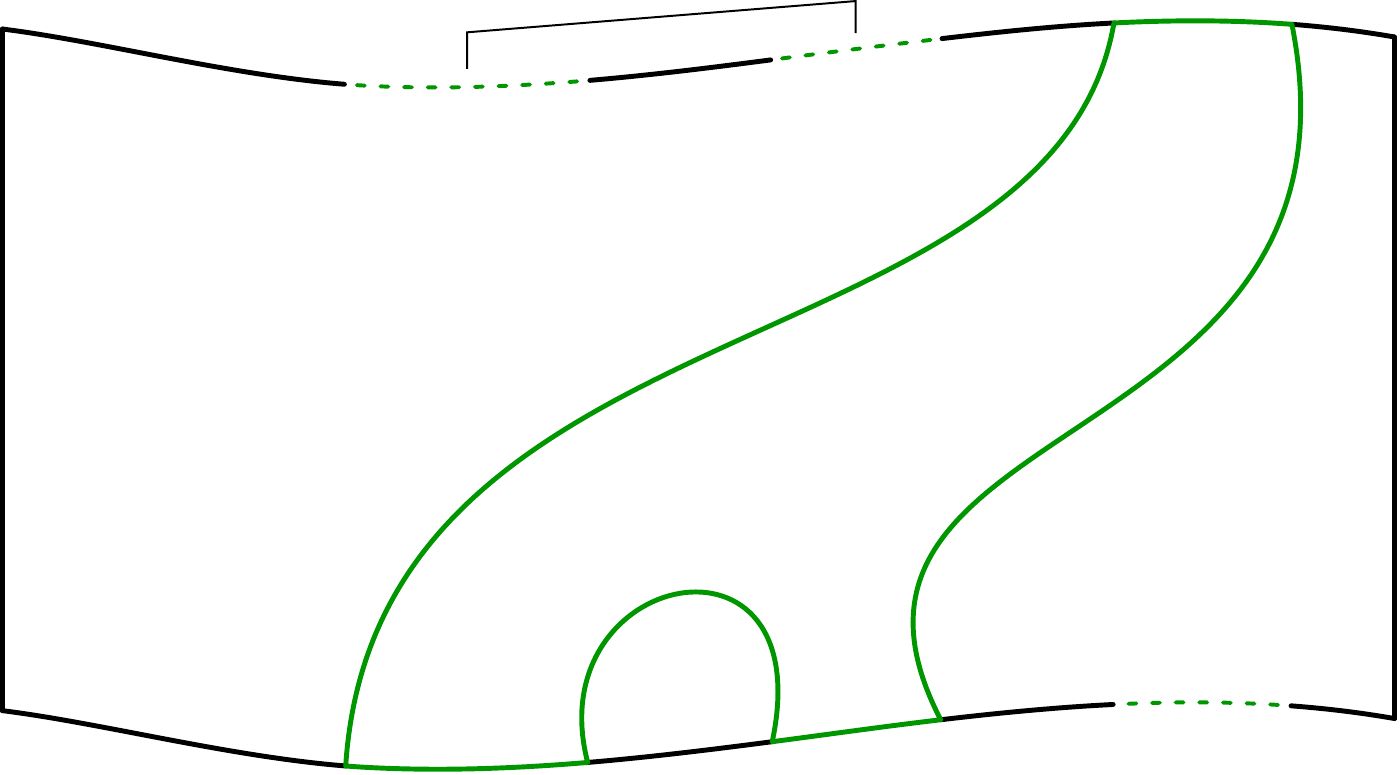}
\put(83,-2){$q^x_{f^D_v}$}
\put(44,61.5){$q^y_{f^D_v}$}
\put(60,25){$f^D_v$}
\end{overpic}
\end{subfigure}		
\caption{a face $f^D_v$, for some $v\in V(G)$, and sets $q^x_{f^D_v}$ and  $q^y_{f^D_v}$. Removing $v$ from $G$ corresponds in $D$ to contracting vertices of $f^D_v$, $q^x_{f^D_v}$ and  $q^y_{f^D_v}$ in three distinct vertices.}
\label{fig:xy_face}
\end{figure}

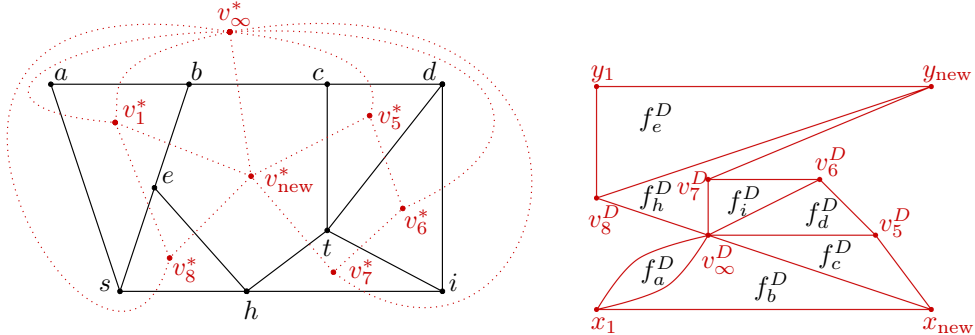
\begin{figure}[h]
\captionsetup[subfigure]{justification=centering}
\centering
\begin{subfigure}[t]{7cm}
\begin{overpic}[width=7cm,percent]{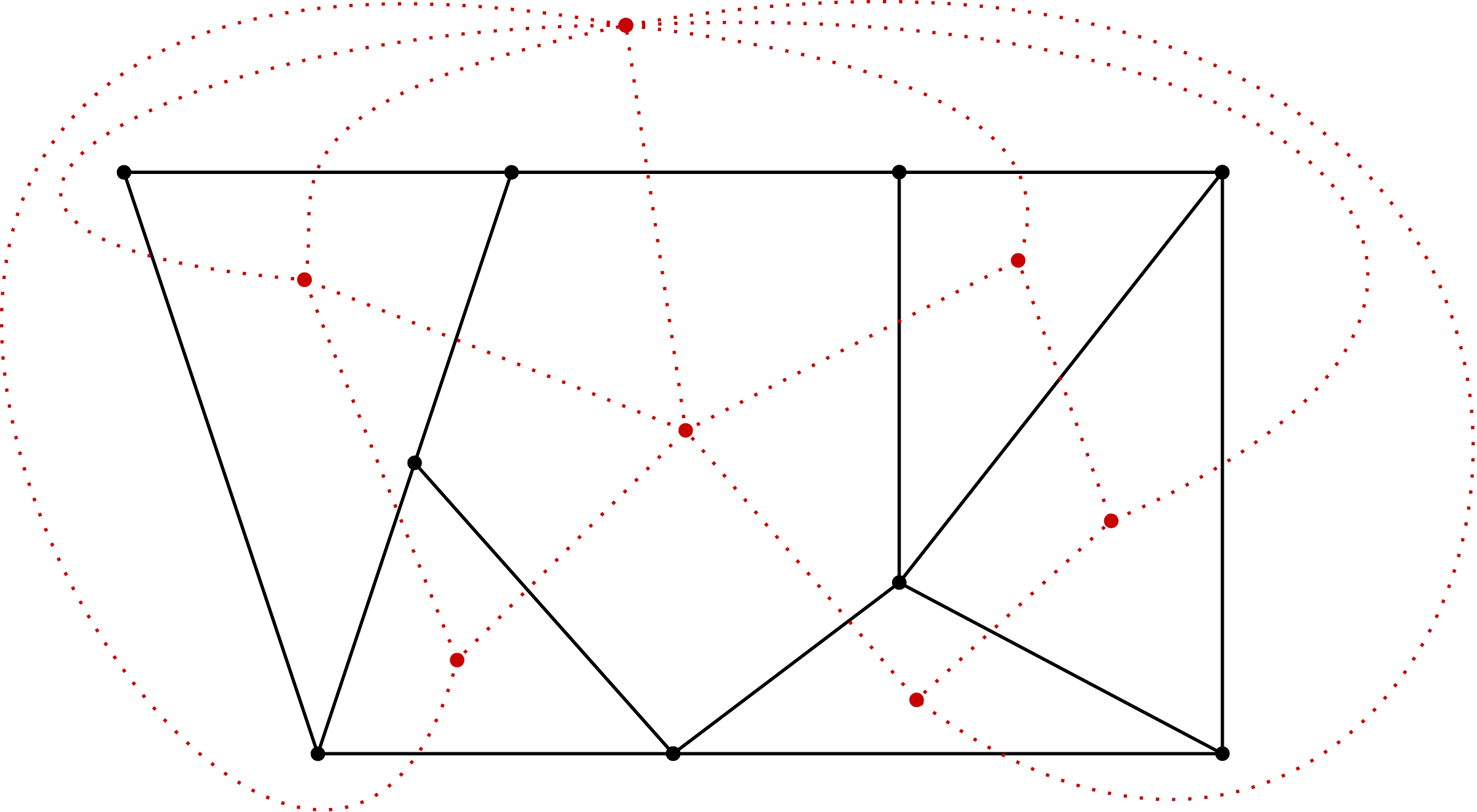}

\put(49,24){\color[rgb]{0.784,0,0}$v^*_{\text{new}}$}

\put(22,37){\color[rgb]{0.784,0,0}$v_1^*$}
\put(70.5,36){\color[rgb]{0.784,0,0}$v_5^*$}
\put(75.5,16){\color[rgb]{0.784,0,0}$v_6^*$}
\put(65,6.5){\color[rgb]{0.784,0,0}$v_7^*$}
\put(31.5,7){\color[rgb]{0.784,0,0}$v_8^*$}
\put(40,55.5){\color[rgb]{0.784,0,0}$v_\infty^*$}

\put(9,44){$a$}
\put(35,44){$b$}
\put(58,44){$c$}
\put(79,44){$d$}
\put(29.5,24.5){$e$}
\put(45,-1.3){$h$}
\put(84,4){$i$}
\put(17.5,4){$s$}
\put(60,10.7){$t$}
\end{overpic}
\end{subfigure}	
\qquad
\begin{subfigure}[t]{4.5cm}
\begin{overpic}[width=4.5cm,percent]{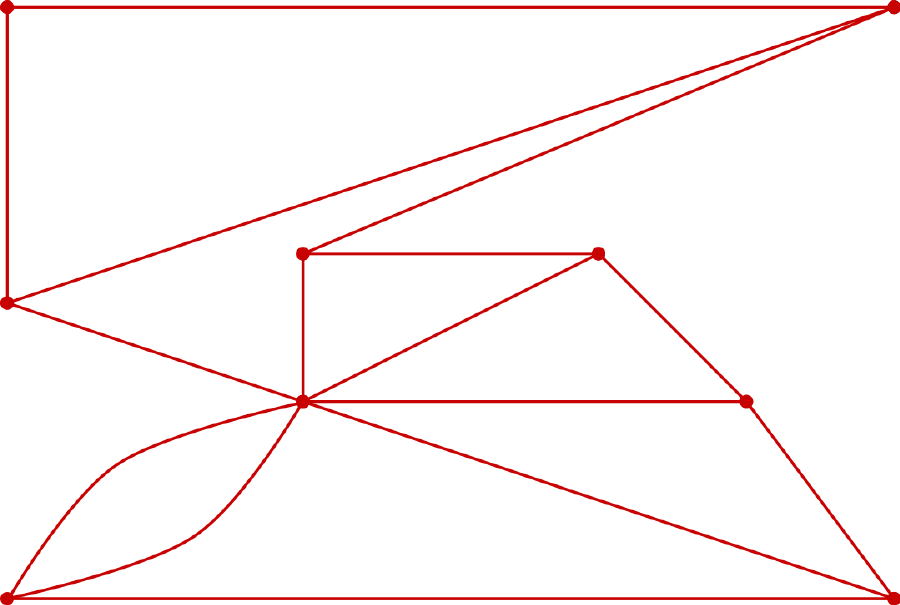}
\put(-1.2,69.6){\color[rgb]{0.784,0,0}$y_1$}
\put(96.8,69.6){\color[rgb]{0.784,0,0}$y_{\text{new}}$}

\put(84,24){\color[rgb]{0.784,0,0}$v^D_5$}
\put(65,42){\color[rgb]{0.784,0,0}$v^D_6$}
\put(24.2,35){\color[rgb]{0.784,0,0}$v^D_7$}
\put(-2,25){\color[rgb]{0.784,0,0}$v^D_8$}
\put(31.5,13.1){\color[rgb]{0.784,0,0}$v^D_\infty$}

\put(-1.2,-5){\color[rgb]{0.784,0,0}$x_1$}
\put(96.8,-5){\color[rgb]{0.784,0,0}$x_{\text{new}}$}

\put(13.5,9.9){${f}^D_a$}
\put(47,5){${f}^D_b$}
\put(66,14.5){${f}^D_c$}
\put(62,27){${f}^D_d$}
\put(13,55){${f}^D_e$}
\put(13,32){${f}^D_h$}
\put(39,30.5){${f}^D_i$}
\end{overpic}
\end{subfigure}				
\caption{starting from graph $G$ in Figure~\ref{fig:G}, we show on the left graph $G-g$ and $(G-g)^*$, and graph $D_{G-g}$ on the right.}
\label{fig:G_setminus_vertex}
\end{figure}

\subsection{Single-crossing $st$-separating cycles}\label{sub:single-crossing}

Itai and Shiloach~\cite{itai-shiloach} consider only shortest $st$-separating cycles that cross $\pi$ exactly once, that correspond in $D$ to paths from $x_i$ to $y_i$, for some $i\in[k]$. Formally, given two paths $p_1,p_2$ in a plane graph, a \emph{crossing} between $p_1$ and $p_2$ is a minimal subpath of $p_1$ defined by vertices $v_1,v_2,\ldots,v_k$, with $k\geq 3$, such that vertices $v_2,\ldots,v_{k-1}$ are contained in $p_2$, and, fixing an orientation of $p_2$, edge $v_1 v_2$ lies to the left of $p_2$ and edge $v_{k-1} v_k$ lies to the right of $p_2$, or vice-versa. We say \emph{$p_1$ crosses $p_2$ $t$ times} if there are $t$ different crossings between $p_1$ and $p_2$.

In our approach, we contract vertices of an edge or a face of $G^*$. Despite this we can still consider only $st$-separating cycles that cross $\pi$ exactly once. The proof of this is the goal of this subsection.

\begin{lemma}\label{lemma:st-separating_unique}
Let $\gamma$ be a simple $st$-separating cycle and let $S$ be either an edge or a face of $G^*$. Let $r=|V(\gamma)\cap V(S)|$. After contracting vertices of $S$ into one vertex, then $\gamma$ becomes the union of $r$ simple cycles and exactly one of them is an $st$-separating cycle.
\end{lemma}
\begin{proof}
Since an edge can be seen as a degenerate face, we prove the statement only in the case in which $S$ is a face $f$. Let $v^*\in V(\gamma)$ and let $u_1^*,u_2^*,\ldots,u_r^*$ be the vertices of $V(\gamma)\cap V(f)$ ordered in clockwise order starting from $v^*$. For convenience, let $u^*_{r+1}=u_1^*$. For $i\in[r]$, let $q_i$ be the clockwise $u_i^* u_{i+1}^*$ path on $\gamma$. After contracting the vertices of $f$ into one, $q_i$ becomes a cycle.  Every $q_i$'s joined with the counterclockwise $u_i^* u_{i+1}^*$ path on the border cycle of $f$ defines a region $R_i$ of $G^*$. 
We remark that if $q_i$ is composed by a single edge $e^*$, then $q_i$ becomes a self-loop and region $R_i$ is a composed only by $e^*$.

Cycle $\gamma$ splits graph $G^*$ into two regions: a region internal to $\gamma$ called $R_{in}$ and an external region called $R_{out}$. W.l.o.g., we assume that $s\in R_{in}$ and $t\in R_{out}$. Now we split the proof into two cases: $f\subseteq R_{in}$ and $f\subseteq R_{out}$. 

\noindent
$\bullet$ Case $f\subseteq R_{in}$. By above, it holds that $R_1,\ldots,R_r\subseteq R_{in}$ (see Figure~\ref{fig:st-separating_internal} on the left). Being $\gamma$ an $st$-separating cycle, then there exists a unique $j\in[r]$ such that $s\in R_j$. Thus, after contracting vertices of $f$ into one, $p_j$ becomes the unique $st$-separating cycle, while all others $R_i$'s become cycles that split $G^*$ into two regions, and each region contains neither $s$ nor $t$ (see Figure~\ref{fig:st-separating_internal} on the right).

\noindent
$\bullet$ Case $f\subseteq R_{out}$. By above there exists a unique $j\in[r]$ such that $R_i\subseteq R_{out}$ for all $i\neq j$ and $R_{in}\subseteq R_j$ (see Figure~\ref{fig:st-separating_external} on the left). W.l.o.g., we assume that $j=r$. After contracting the vertices of $f$ into one, all regions $R_1,\ldots,R_{r-1}$ become regions inside $R_r$ because of the embedding (see Figure~\ref{fig:st-separating_external} on the right). We recall that $s\in R_{in}$, thus there are two cases: if $t\in R_i$ for some $i\in[r-1]$, then $p_i$ becomes the unique $st$-separating cycle; otherwise, $t\in R_{out}$, and thus $p_r$ becomes the unique $st$-separating cycle.\qed
\end{proof}


\begin{figure}[h]
\captionsetup[subfigure]{justification=centering}
\centering
\begin{subfigure}[t]{12cm}
\begin{overpic}[width=12cm,percent]{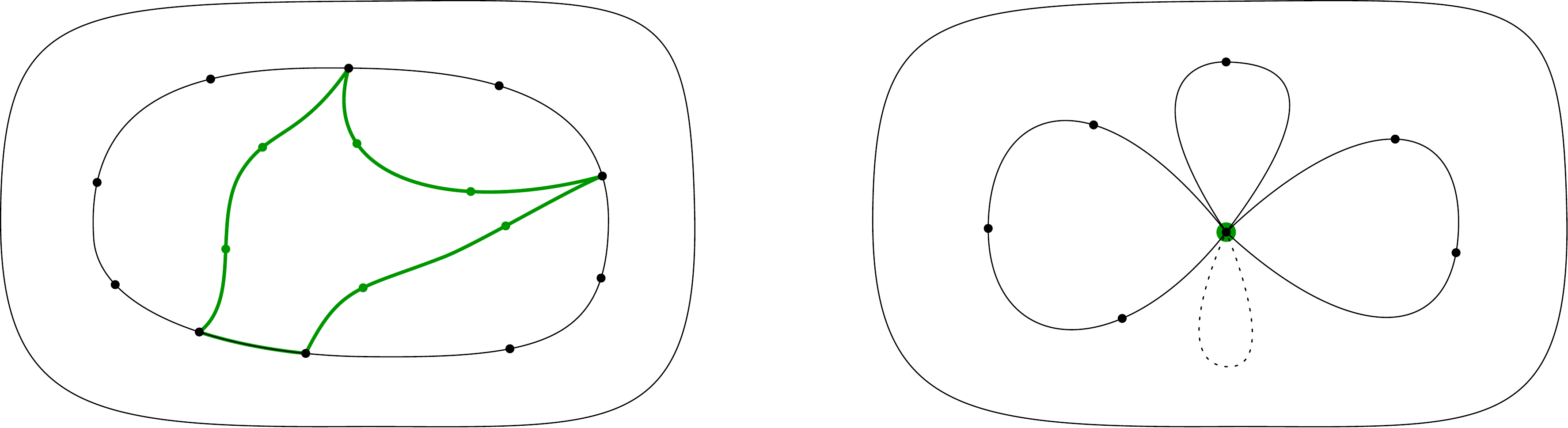}


\put(9,15){$R_2$}
\put(26,18){$R_3$}
\put(29,8){$R_4$}
\put(20,12.5){$f$}
\put(-3,24){$f_\infty$}
\put(6,21){$\gamma$}
\put(24.5,1.2){$R_{out}$}

\put(3.8,6.3){$a_1^*$}
\put(2.3,15){$a_2^*$}
\put(11.5,23.5){$a_3^*$}
\put(32,22.5){$a_4^*$}
\put(39,8){$a_5^*$}
\put(33,3){$a_6^*$}

\put(10,3){$u_4^*$}
\put(22,24){$u_1^*$}
\put(39,16){$u_2^*$}
\put(18,1.5){$u_3^*$}


\put(68,13){$R_2$}
\put(76.5,18){$R_3$}
\put(85,12){$R_4$}
\put(86,2){$R_{out}$}
\put(52.7,24){$f_\infty$}
\put(62.1,18){$\gamma$}

\put(71,4.5){$a_1^*$}
\put(59,12){$a_2^*$}
\put(68,20.5){$a_3^*$}
\put(77.5,24.2){$a_4^*$}
\put(88,19.5){$a_5^*$}
\put(93.5,10.5){$a_6^*$}

\end{overpic}
\end{subfigure}			
\caption{on the left, cycle $\gamma$ and face $f$ belonging to $R_{in}$. On the right, cycle $\gamma$ after contracting all vertices of $f$ into one, the dashed edge represents a self-loop.}
\label{fig:st-separating_internal}
\end{figure}

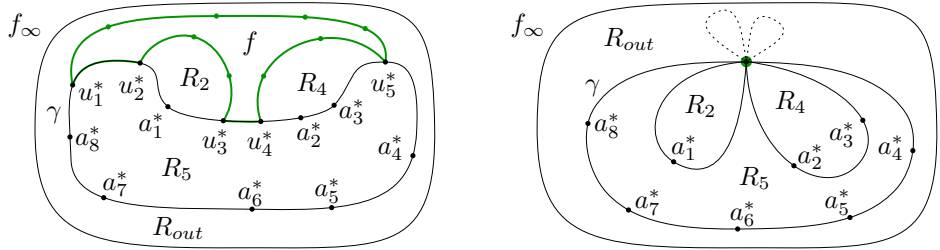
\begin{figure}[h]
\captionsetup[subfigure]{justification=centering}
\centering
\begin{subfigure}[t]{12cm}
\begin{overpic}[width=12cm,percent]{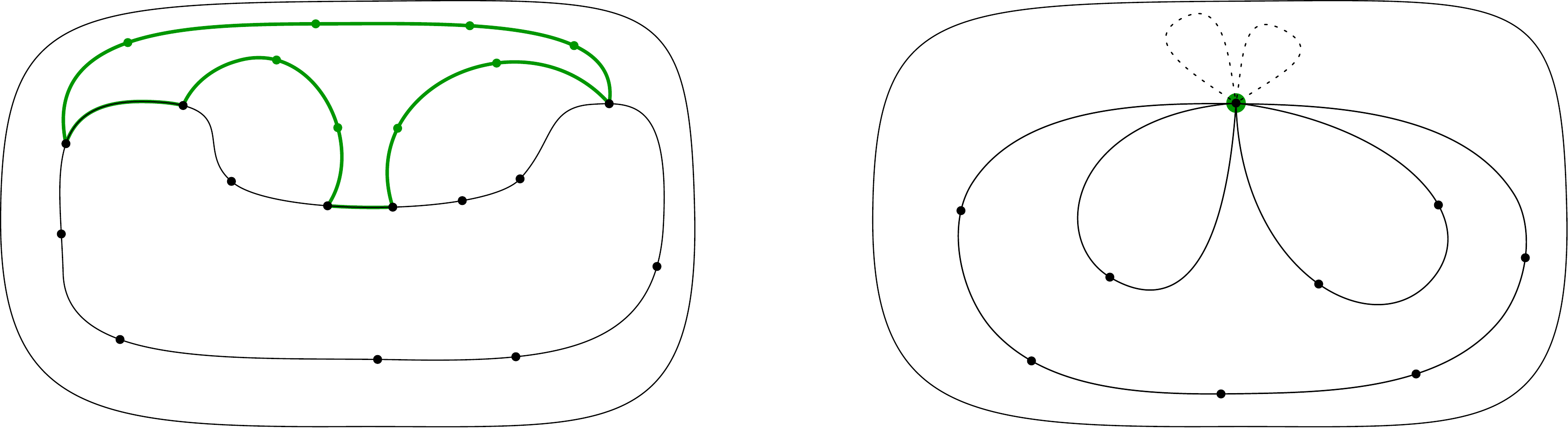}
\put(23,22){$f$}
\put(16,18){$R_2$}
\put(29,17.5){$R_4$}
\put(14,8){$R_5$}
\put(13,1.5){$R_{out}$}
\put(-3,24){$f_\infty$}
\put(1.3,14.5){$\gamma$}

\put(11.5,13){$a_1^*$}
\put(29,12){$a_2^*$}
\put(33.7,14.5){$a_3^*$}
\put(38,10.5){$a_4^*$}
\put(31,5.8){$a_5^*$}
\put(22.5,5.5){$a_6^*$}
\put(7.5,6.5){$a_7^*$}
\put(4.5,11.5){$a_8^*$}

\put(4.8,16.5){$u_1^*$}
\put(9.3,17.2){$u_2^*$}
\put(18.5,11.2){$u_3^*$}
\put(23.5,11){$u_4^*$}
\put(37.2,17.5){$u_5^*$}


\put(72,15){$R_2$}
\put(82,15.5){$R_4$}
\put(78,7){$R_5$}
\put(63,22.5){$R_{out}$}
\put(52.7,24){$f_\infty$}
\put(61,17.5){$\gamma$}

\put(70.5,10.5){$a_1^*$}
\put(84.5,9.5){$a_2^*$}
\put(88,12){$a_3^*$}
\put(93.4,10.5){$a_4^*$}
\put(87.5,4.5){$a_5^*$}
\put(77,3.2){$a_6^*$}
\put(66.5,4.5){$a_7^*$}
\put(62,13){$a_8^*$}
\end{overpic}
\end{subfigure}			
\caption{on the left, cycle $\gamma$ and face $f$ belonging to $R_{out}$. On the right, cycle $\gamma$ after contracting all vertices of $f$ into one, dashed edges represent self-loops.}
\label{fig:st-separating_external}
\end{figure}

Let $\Gamma$ be the set of all $st$-separating cycles in $G^*$, and let $\Gamma_1$ be the set of all $st$-separating cycles in $G^*$ that cross $\pi$ exactly once.
Given $\gamma\in\Gamma$ and either an edge or a face $S$ of $G^*$, thanks to Lemma~\ref{lemma:st-separating_unique} we can define $\Delta_S(\gamma)$ as ``the length of the unique $st$-separating cycle contained in $\gamma$ after contracting vertices of $S$ into one''. Being $\MF$ equal to the length of a minimum $st$-separating cycle, the following relations hold:

\begin{equation*}
\text{for any $e\in E(G)$, $vit(e)=\MF-\min_{\gamma\in\Gamma}\Delta_{e^*}(\gamma)$,}
\end{equation*}
\begin{equation*}
\text{for any $v\in V(G)\setminus\{s,t\}$, $vit(v)=\MF-\min_{\gamma\in\Gamma}\Delta_{f^*_v}(\gamma)$}.
\end{equation*}

Now we show that in the above equations we can replace set $\Gamma$ with set $\Gamma_1$.

\begin{lemma}\label{lemma:multicrossing}
Let $e\in E(G)$ and $v\in V(G)\setminus\{s,t\}$. It holds that $vit(e)=\MF-\min_{\gamma\in\Gamma_1}\Delta_{e^*}(\gamma)$ and $vit(v)=\MF-\min_{\gamma\in\Gamma_1}\Delta_{f^*_v}(\gamma)$.
\end{lemma}
\begin{proof}
We recall that removing an edge $e$ from $G$ corresponds to contracting endpoints of $e^*$ into one vertex, while removing a vertex $v$ from $G$ corresponds to contracting all the vertices in face $f^*_v$ into one vertex. So we prove the thesis only in the more general case of vertex removal. For convenience, we denote $f^*_v$ by $f$. Let $\gamma\in\Gamma$ be such that $vit(f)=\MF-\Delta_f(\gamma)$ and assume that $\gamma\not\in\Gamma_1$. If $V(\gamma)\cap V(f)=\emptyset$, then $vit(f)=0$, hence it suffices to remove crossings between $\gamma$ and $\pi$, see~\cite{itai-shiloach}. 
Thus let us assume that $V(\gamma)\cap V(f)\neq\emptyset$.

By Lemma~\ref{lemma:st-separating_unique}, there exist unique $a^*,b^*\in V(f)\cap V(\gamma)$ such that the clockwise $a^* b^*$ path $p$ on $\gamma$ becomes an $st$-separating cycle after the contraction of vertices of $f$ into one. Then we remove crossing between $p$ and $\pi$ in order to obtain a path $p'$ not longer than $p$ as above. Finally, let $\gamma'=p\circ q$, where $q$ is the clockwise $a^* b^*$ path on $f$. It holds that $\gamma'\in\Gamma_1$ and $\Delta_f(\gamma')\leq\Delta_f(\gamma)$, the thesis follows.\qed
\end{proof}

\subsection{Vitality vs.\ distances in $D$}\label{sub:what_we_need}

The main results of this subsection are Proposition~\ref{prop:vitality_edges} and Proposition~\ref{prop:vitality_vertices}. The first proposition shows which distances in $D$ are needed to obtain edge vitality and in the latter proposition we do the same for vertex vitality. 
In Subsection~\ref{sub:effects_deleting_edge_or_vertex} we have proved that removing an edge or a vertex from $G$ corresponds to contracting in single vertices some sets of vertices of $D$. The main result of Proposition~\ref{prop:vitality_edges} and Proposition~\ref{prop:vitality_vertices} is that we can consider these vertices individually. 

Let $e$ be an edge of $G$. The removal of $e$ from $G$ corresponds to the contraction of endpoints of $e^*$ into one vertex in $G^*$. Thus if an $st$-separating cycle $\gamma$ of $G^*$ contains $e^*$, then the removal of $e$ from $G$ reduces the length of $\gamma$ by $w(e^*)$. Thus $e$ has strictly positive vitality if and only if there exists an $st$-separating cycle $\gamma$ in $G^*$ whose length is strictly less than $\MF+w(e^*)$ and $e^*\in \gamma$. This is the main idea to compute the vitality of all edges. Now we have to translate it to ${D}$.

We observe that capacities of edges in $G$ become lengths (or weights) in ${D}$. For this reason, we define $w(e^D)=c(e)$, for all edges $e\in G$ satisfying $e^*\not\in\pi$ and $w(e^D_x)=w(e^D_y)=c(e)$ for all edges $e\in G$ satisfying $e^*\in\pi$.

For $i\in[k]$, we define $d_i=\dist_{D}(x_i,y_i)$. We observe that $\MF=\min_{i\in[k]}d_i$.
For a subset $S$ of $V({D})$ and any $i\in[k]$ we define $d_i(S)=\min\{d_i,\dist_{D}(x_i,S)+\dist_{D}(y_i,S)\}$. We observe that $d_i(S)$ represents the distance in ${D}$ from $x_i$ to $y_i$ if all vertices of $S$ are contracted into one. 

For every $x\in V(G)\cup E(G)$ we define $\MF_x$ as the max flow in graph $G-x$. By definition, $vit(x)=\MF-\MF_x$ and, trivially, $x$ has strictly positive vitality if and only if $\MF_x<\MF$.

\begin{proposition}\label{prop:vitality_edges}
For every edge $e$ of $G$, if $e^*\not\in \pi$, then $\MF_e=\min_{i\in[k]}\{d_i(e^D)\}$. If $e^*\in \pi$, then
$\MF_e=\min_{i\in[k]}\big\{\min\{d_i(e^D_x),d_i(e^D_y)\}\big\}$.
\end{proposition}
\begin{proof} 
Let $e$ be an edge of $G$. If $vit(e)=0$, then $\MF_e=\MF$ and the thesis trivially holds. Hence let us assume $vit(e)>0$, then Lemma~\ref{lemma:multicrossing} there exists an $st$-separating cycle in $G^*$ that crosses $\pi$ exactly once satisfying $w(\gamma)<\MF+w(e^*)$ and $e^*\in \gamma$. If $e^*\not\in\pi$, then $e$ corresponds in $D$ to  edge $e^D$, thus the thesis holds. If $e^*\in\pi$, then we note that every path in $D$ containing both $e^D_x$ and $e^D_y$ corresponds in $G^*$ to an $st$-separating cycle that passes through $e^*$ twice, thus its length is equal or greater than $\MF+2 c(e)$. Thus we consider only paths that contain $e^D_x$ or $e^D_y$ but not both. The thesis follows.\qed
\end{proof}

Note that if $f^*_v$ and $\pi$ have some common vertices, then one among $q^x_{f^D_v}$ and $q^y_{f^D_v}$ could be empty. For this reason, we set $d_i(\emptyset)=+\infty$, for all $i\in[k]$. 

\begin{proposition}\label{prop:vitality_vertices}
For every vertex $v$ of $G$, if $f^*_v$ and $\pi$ have no common vertices, then $\MF_v=\min_{i\in[k]}\{d_i(f)\}$, where $f=f^D_v$, otherwise
\begin{equation}\label{eq:vitality_vertices}
\MF_v=\min\left\{\begin{array}{l}
\min_{i\in[k]}\{d_i(f)\},\\
\min_{i\in[k]}\{d_i(q^x_f)\},\\
\min_{i\in[k]}\{d_i(q^y_f)\},\\
\dist_{D}(f,q^x_f),\\
\dist_{D}(f,q^y_f)
\end{array}\right\}.
\end{equation}
\end{proposition}
\begin{proof}
If $f^*_v$ and $\pi$ have no common vertices, then the proof is analogous to the edge case. Thus let us assume that $f^*_v$ and $\pi$ have common vertices. 
Let $D'$ be the graph obtained from $D$ by adding a vertex $u,v,z$ connected with all vertices of $q^x_f$, of $q^y_f$, of $f$, respectively, with zero weight edges; for convenience we assume that $q^x_f$ and $q^y_f$ are both not empty. By Lemma~\ref{lemma:multicrossing} and discussion in Subsection~\ref{sub:effects_deleting_edge_or_vertex}, $\MF_v=w(p)$, where $p$ is a shortest $x_i y_i$ path in $D'$, varying $i\in[k]$. 

Note that after contracting vertices of $f$ into one vertex there exists an $x_iy_i$ path whose length is $dist_D(f,x_i)$, for all $x_i\in q^x_f$. In particular, there exists an  $x_iy_i$ path whose length is $dist_D(f,q^x_f)$, for some $i$ satisfying $x_i\in q^x_f$. The same argument applies for $q^y_f$. This implies that if $vit(v)=0$, then Equation~\eqref{eq:vitality_vertices} is correct. Hence we assume that $vit(v)>0$, so at least one among $u,v$ and $z$ belongs to $p$.

If $u\in p$ and $v,z\not\in p$ (resp.,  $v\in p$ and $u,z\not\in p$), then $w(p)=\min_{i\in[k]}\{d_i(q^x_f)\}$ (resp., $w(p)=\min_{i\in[k]}\{d_i(q^y_f)\}$) . If $z\in p$ and $u,v\not\in p$ then $w(p)=\min_{i\in[k]}\{d_i(f)\}$. We have analyzed all cases in which $p$ contains exactly one vertex among $u,v$ and $z$. To complete the proof, we prove that, for any $i\in[k]$, every $x_iy_i$ path that contains at least two vertices among $u,v$ and $z$ also contains a subpath whose length is at least $\min\{\dist_{D}(f,q^x_f),\dist_{D}(f,q^y_f)\}$.

Let $\ell$ be an $x_iy_i$ path, for some $i\in[k]$. If $u,z\in \ell$, then there exists a subpath $\ell'$ of $\ell$ from a vertex $x_j$ of $q^x_f$ to a vertex $r$ of $f$. If we add to $\ell'$ the two zero weigthed edges $rz$ and $zy_j$ we obtain a $x_jy_j$ path whose length is at least $\dist_{D}(f,q^x_f)$. We can use a symmetric strategy if $v,z\in \ell$. 

It remains only the case in which $u,v\in\ell$. If $q^x_f$ and $q^y_f$ are both non-empty, then $f$ splits $D$ and $D'$ into two or more parts and no part contains vertices of both $q^x_f$ and $q^y_f$ (see Figure~\ref{fig:xy_face}). Thus if $u,v\in \ell$, then $\ell$ passes through at least one vertex of $f$, implying that $\ell$ has a subpath from a vertex of $f$ to a vertex of $q^x_f$, or $q^y_f$. As above, this path can be transformed in a $x_jy_j$ path shorter than $\ell$ whose length is at least $\min\{\dist_{D}(f,q^x_f),\dist_{D}(f,q^y_f)\}$, for some $j\in[k]$.\qed
\end{proof}

\section{Slicing graph $D$ preserving approximated distances}\label{sec:divide_and_conquer}
In this section we explain our divide and conquer strategy. We slice graph $D$ along shortest $x_iy_i$'s paths. If these paths have lengths that differ at most $\delta$, then we have a $\delta$ additive approximation of distances required in Proposition~\ref{prop:vitality_edges} and Proposition~\ref{prop:vitality_vertices} by looking into a single slice instead of the whole graph $D$. This result is stated in Lemma~\ref{lemma:Omega_i}. These slices can share boundary vertices and edges, implying that their dimension might be $O(n^2)$. In Lemma~\ref{lemma:Henzinger_in_Omega_i} we compute an implicit representation of these slices in linear time.

From now on, we mainly work on graph $D$, thus we omit the superscript $D$ unless we refer to $G$ or $G^*$. To work in ${D}$ we need a shortest $x_iy_i$ path and its length, for all $i\in[k]$. In the following theorem we show time complexities for obtaining elements in ${D}$. We say that two paths are \emph{single-touch} if their intersection is still a path.

Given two graphs $A=(V(A),E(A))$ and $B=(V(B),E(B))$ we define $A\cup B=(V(A)\cup V(B),E(A)\cup E(B))$ and $A\cap B=(V(A)\cap V(B),E(A)\cap E(B))$. 

\begin{theorem}[\cite{err_giappo},\cite{gabow-tarjan},\cite{italiano}]\label{th:U_italiano+err_giappo}
If $G$ is a positive edge-weighted planar graph,
\begin{itemize}\itemsep0em
\item we compute $U=\bigcup_{i\in[k]}p_i$ and $w(p_i)$ for all $i\in[k]$, where $p_i$ is a shortest $x_iy_i$ path in ${D}$ and $\{p_i\}_{i\in[k]}$ is a set of pairwise non-crossing single-touch paths, in $O(n\log\log n)$ time---see~\cite{italiano} for computing $U$ and~\cite{err_giappo} for computing $w(p_i)$'s, 
\item  for every $I\subseteq[k]$, we compute $\bigcup_{i\in I}p_i$ in $O(n)$  time---see~\cite{gabow-tarjan} by noting that $U$ is a forest and the paths can be found by using nearest common ancestor queries.
\end{itemize}
\end{theorem}

From now on, for each $i\in[k]$ we fix a shortest $x_iy_i$ path $p_i$, and we assume that $\{p_i\}_{i\in[k]}$ is a set of pairwise single-touch non-crossing shortest paths. Let $U=\bigcup_{i\in[k]}p_i$, see Figure~\ref{fig:Left_Right_and_Omega_i}(\subref{fig:U}). 

Given an $ab$ path $p$ and a $bc$ path $q$, we define $p\circ q$ as the (possibily not simple) $ac$ path obtained by the union of $p$ and $q$. Each $p_i$'s splits ${D}$ into two parts as shown in the following definition and in Figure~\ref{fig:Left_Right_and_Omega_i}(\subref{fig:Left_Right}).

\begin{definition}\label{def:Left_and_Right}
For every $i\in[k]$, we define $\Left_i$ as the subgraph of ${D}$ bounded by the cycle $\pi_y[y_1,y_i]\circ p_i\circ \pi_x[x_i,x_1]\circ l$, where $l$ is the leftmost $x_1y_1$ path in ${D}$. Similarly, we define $\Right_i$ as the subgraph of ${D}$ bounded by the cycle $\pi_y[y_i,y_k]\circ r \circ \pi_x[x_k,x_i]\circ p_i$, where $r$ is the rightmost $x_ky_k$ path in ${D}$.
\end{definition}

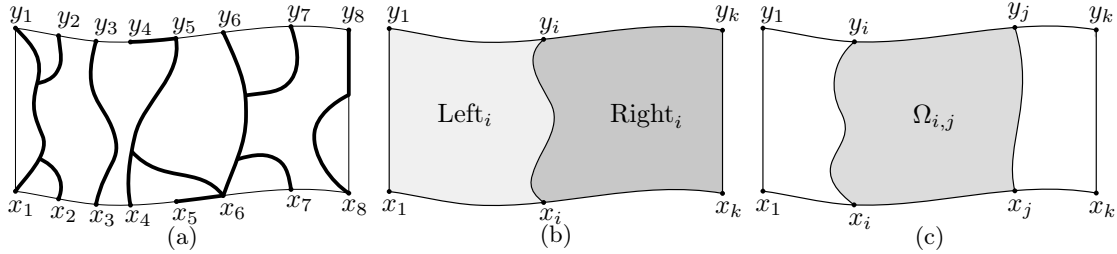
\begin{figure}[h]
\captionsetup[subfigure]{justification=centering}
\centering
\begin{subfigure}{4.5cm}
\begin{overpic}[width=4.5cm,percent]{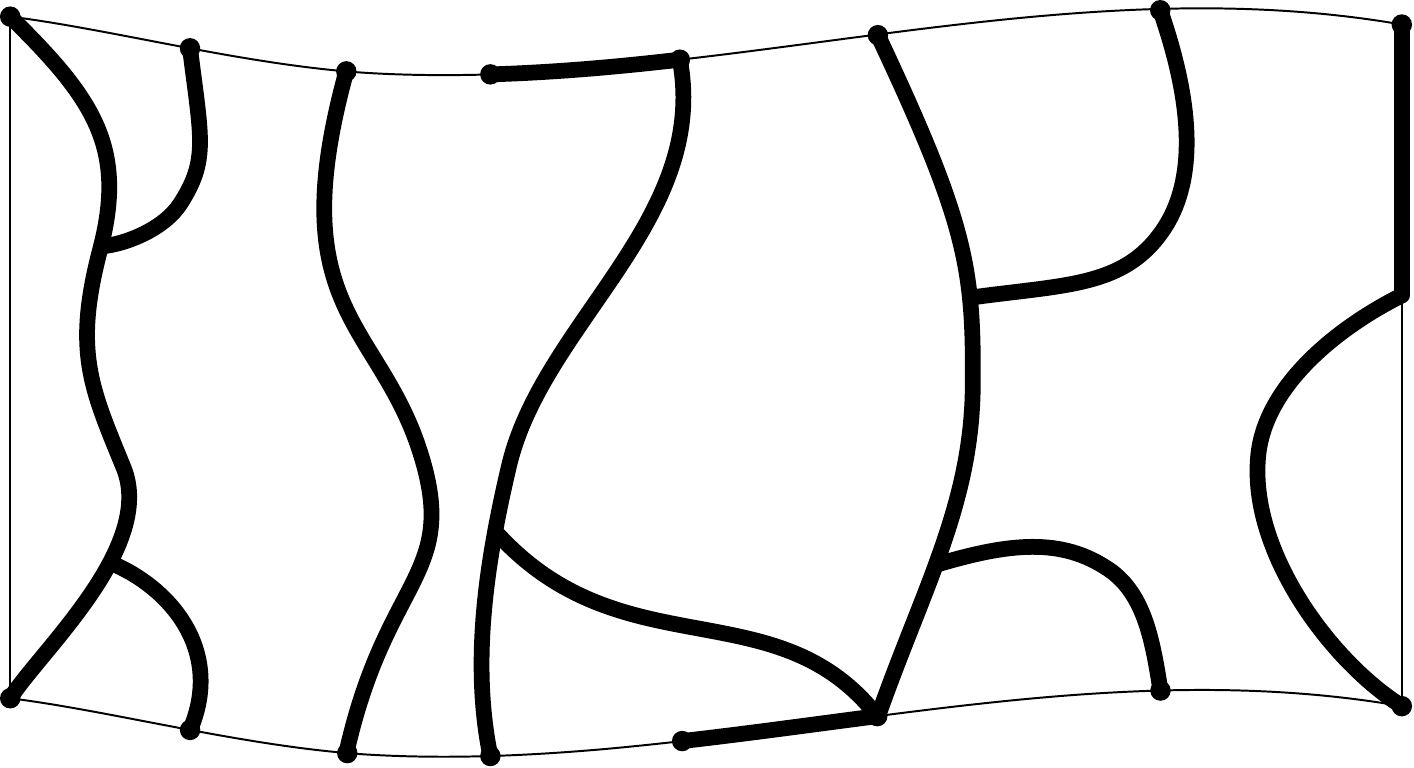}

\put(-1.5,-1.2){$x_1$}
\put(-1,56){$y_1$}
\put(11,-2.9){$x_2$}
\put(13,54.5){$y_2$}
\put(22.5,-4.2){$x_3$}
\put(24,52.5){$y_3$}
\put(33,-4.2){$x_4$}
\put(34,52.2){$y_4$}
\put(47,-3.5){$x_5$}
\put(47,53){$y_5$}
\put(61,-2){$x_6$}
\put(60.5,55){$y_6$}
\put(80.7,0.2){$x_7$}
\put(81,56.5){$y_7$}
\put(97,-1){$x_8$}
\put(97,56){$y_8$}
\end{overpic}
\caption{}\label{fig:U}
\end{subfigure}
\quad
\begin{subfigure}{4.5cm}
\begin{overpic}[width=4.5cm,percent]{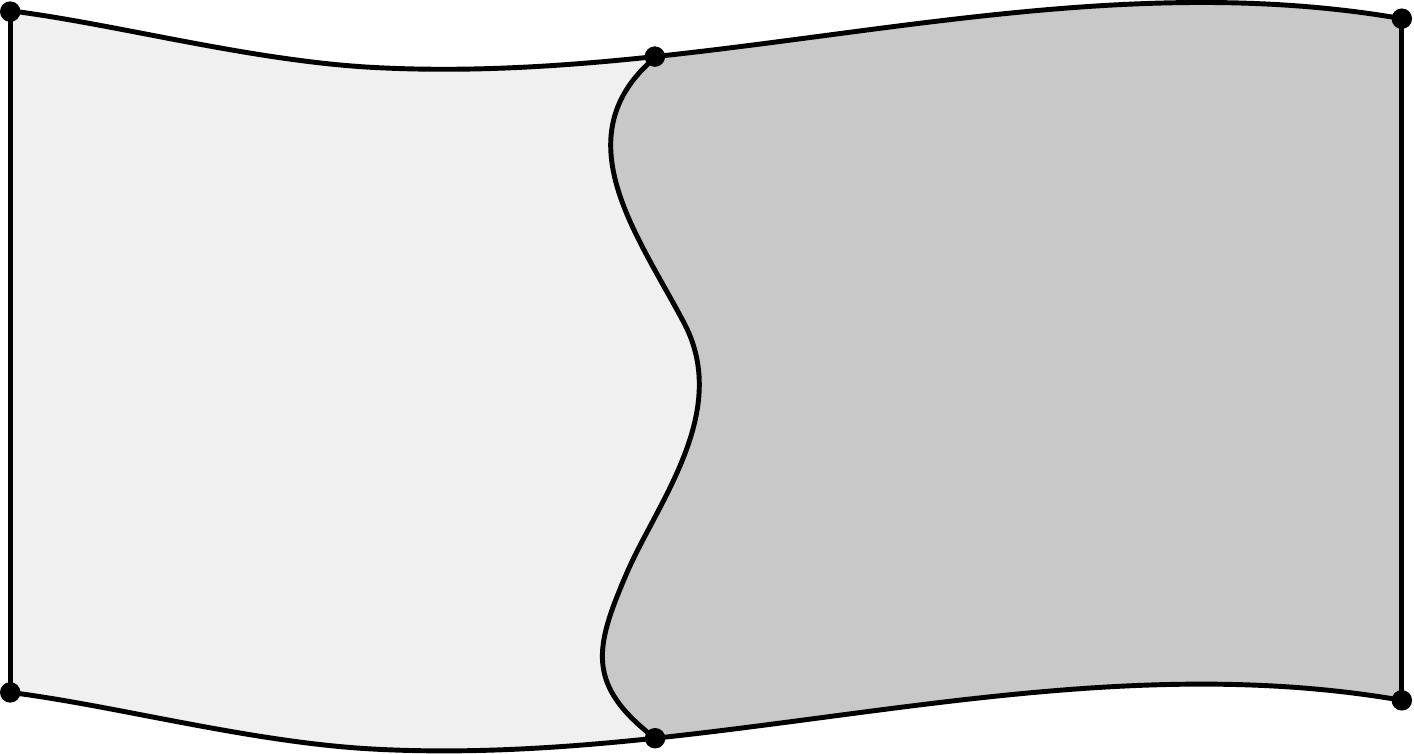}
\put(-1.5,-1.5){$x_1$}
\put(0,55.5){$y_1$}
\put(45,-4){$x_i$}
\put(45,52.5){$y_i$}
\put(97,-2){$x_k$}
\put(97,56){$y_k$}

\put(15,25){$\Left_i$}
\put(66,25){$\Right_i$}
\end{overpic}
\caption{}\label{fig:Left_Right}
\end{subfigure}		
\quad
\begin{subfigure}{4.5cm}
\begin{overpic}[width=4.5cm,percent]{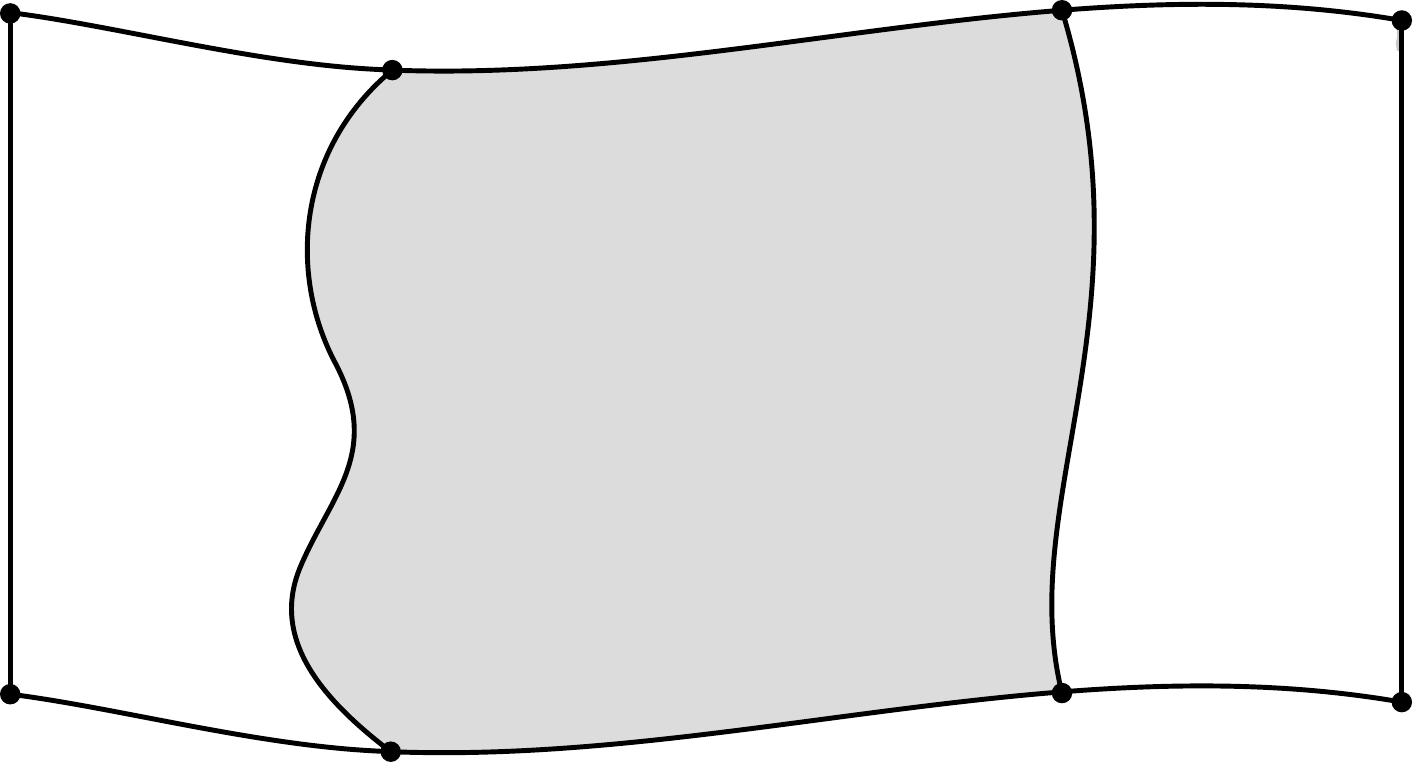}
\put(-1.5,-1){$x_1$}
\put(0,56){$y_1$}

\put(26,-5){$x_i$}
\put(26,52.5){$y_i$}

\put(73,-1){$x_j$}
\put(73,57.3){$y_j$}

\put(97,-1.5){$x_k$}
\put(97,56){$y_k$}

\put(45,25){$\Omega_{i,j}$}
\end{overpic}
\caption{}\label{fig:Omega_i,j}
\end{subfigure}		
\caption{in (\subref{fig:U}) the graph $U$ in bold and in (\subref{fig:Left_Right}) subgraphs $\Left_i$ and $\Right_i$ are highlighted. In (\subref{fig:Omega_i,j}) subgraph $\Omega_{i,j}$, for some $i<j$.}
\label{fig:Left_Right_and_Omega_i}
\end{figure}

Based on Definition~\ref{def:Left_and_Right}, for every $i,j\in[k]$, with $i<j$, we define $\Omega_{i,j}=\Right_i\cap\Left_j$, see Figure~\ref{fig:Left_Right_and_Omega_i}(\subref{fig:Omega_i,j}).  We classify $(x_i,y_i)$'s pairs according to the difference between $d_i$ and $\MF$. Each class contains pairs for which this difference is about $r$ times $\delta$; where $\delta>0$ is an arbitrarily fixed value.

For every $r\in\mathbb{N}$, we define $L_r=(\ell_1^r,\ldots,\ell_{z_r}^r)$ as the ordered list of indices in $[k]$ such that $d_{j}\in[\MF+\delta r,\MF+\delta(r+1))$, for all $j\in L_r$, and $\ell_j^r<\ell_{j+1}^r$ for all $j\in[z_r-1]$. It is possible that $L_r=\emptyset$ for some $r>0$ (it holds that $L_0\neq\emptyset$). If no confusion arises, we omit the superscript $r$; thus we write $\ell_i$ in place of $\ell_i^r$.

The following lemma is the key of our slicing strategy. 
In particular, Lemma~\ref{lemma:Omega_i} can be applied for computing distances required in Proposition~\ref{prop:vitality_edges} and Proposition~\ref{prop:vitality_vertices}, since the vertex set of a face or an edge of $D$ is always contained in a slice. An application is in Figure~\ref{fig:application_alfabeto_greco}.

\begin{lemma}\label{lemma:Omega_i}
Let $r>0$ and let $L_r=(\ell_1,\ell_2,\ldots,\ell_z)$. Let $S$ be a set of vertices of ${D}$ with $S\subseteq \Omega_{\ell_i,\ell_{i+1}}$ for some $i\in[z-1]$. Then
\begin{equation*}
\min_{\ell\in L_r}d_\ell(S)>\min\{{d_{\ell_{i}}(S),d_{\ell_{i+1}}(S)}\}-\delta.
\end{equation*}
Moreover, if $S\subseteq\Left_{\ell_1}$ (resp., $S\subseteq\Right_{\ell_z}$) then $\min_{\ell\in L_r}d_\ell(S)> d_{\ell_1}(S)-\delta$ (resp., $\min_{\ell\in L_r}d_\ell(S)> d_{\ell_z}(S)-\delta$).
\end{lemma}
\begin{proof}
We need the following crucial claim.
\begin{enumerate}[label=\alph{CONT})]\stepcounter{CONT}
\item\label{claim:left_right} Let $i<j\in L_r$. Let $L$ be a set of vertices in $\Left_i$ and let $R$ be set of vertices in $\Right_j$. Then $d_i(L)< d_j(L)+\delta$ and $d_j(R)< d_i(R)+\delta$.
\end{enumerate}
\claimbegin{\ref{claim:left_right}} 
we prove that $d_i(L)< d_j(L)+\delta$. By symmetry, it also proves that $d_j(R)< d_i(R)+\delta$. Let us assume by contradiction that $d_i(L)\geq d_j(L)+\delta$.


Let $\alpha$ (resp., $\epsilon$, $\mu$, $\nu$) be a path from $x_i$ (resp., $y_i$, $x_j$, $y_j$) to $z_\alpha$ (resp., $z_\epsilon$, $z_\mu$, $z_\nu$) whose length is $d(x_i,L)$ (resp. $d(y_i,L)$, $d(x_j,L)$, $d(y_j,L)$), see Figure~\ref{fig:alfabeto_greco} on the left. Being $x_j,y_j\in\Right_i$ and $L\subseteq\Left_i$, then $\mu$ and $\nu$ cross $p_i$. Let $v$ be the vertex that appears first in $p_i\cap \mu$ starting from $x_j$ on $\mu$ and let $u$ be the vertex that appears first in $p_i\cap \nu$ starting from $y_j$ on $\nu$. An example of these paths is in Figure~\ref{fig:alfabeto_greco} on the left. Let $\zeta=p_i[y_i,u]$, $\theta=p_i[u,v]$, $\beta=p_i[x_i,v]$, $\kappa=\mu[x_j,v]$, $\iota=\nu[y_j,u]$, $\eta=\nu[u,z_\nu]$ and $\gamma=\mu[v,z_\mu]$, see Figure~\ref{fig:alfabeto_greco} on the right.

Now $\cc(\beta)+\cc(\gamma)\geq \cc(\alpha)$, otherwise $\alpha$ would not be a shortest path from $x_i$ to $L$. Similarly $\cc(\zeta)+\cc(\eta)\geq \cc(\epsilon)$. Moreover, being $\cc(\zeta)+\cc(\theta)+\cc(\beta)=d_i$, then $\cc(\theta)\leq d_i-\cc(\alpha)+\cc(\gamma)-\cc(\epsilon)+\cc(\eta)$. Being $d_i(L)\geq d_j(L)+\delta$, then $\cc(\alpha)+\cc(\epsilon)\geq \cc(\mu)+\cc(\nu)+\delta$, this implies $\cc(\alpha)+\cc(\epsilon)\geq\cc(\kappa)+\cc(\gamma)+\cc(\iota)+\cc(\eta)+\delta$. 
 
It holds that $\cc(\theta)+\cc(\kappa)+\cc(\iota)\leq d_i-\cc(\alpha)+\cc(\gamma)-\cc(\epsilon)+\cc(\eta)+\cc(\alpha)+\cc(\epsilon)-\cc(\gamma)-\cc(\eta)-\delta=d_i-\delta< d_j$ because $i,j\in L_r$ imply $|d_i-d_j|<\delta$. Thus $\kappa\circ \theta\circ \iota$ is a  path from $x_j$ to $y_j$ strictly shorter than $d_j$, absurdum.
\claimend{\ref{claim:left_right}}

Being $S\subseteq\Right_{\ell_j}$ for all $j<i$ and $S\subseteq\Left_{\ell_{j'}}$ for all $j'>i+1$, then the first part of the thesis follows from~\ref{claim:left_right}. The second part follows also from~\ref{claim:left_right} by observing that if $S\subseteq\Left_{\ell_1}$, then $S\subseteq\Left_{\ell_i}$ for all $i\in L_r$.\qed
\end{proof}

\begin{figure}[h]
\captionsetup[subfigure]{justification=centering}
\centering
	\begin{subfigure}{5cm}
\begin{overpic}[width=5cm,percent]{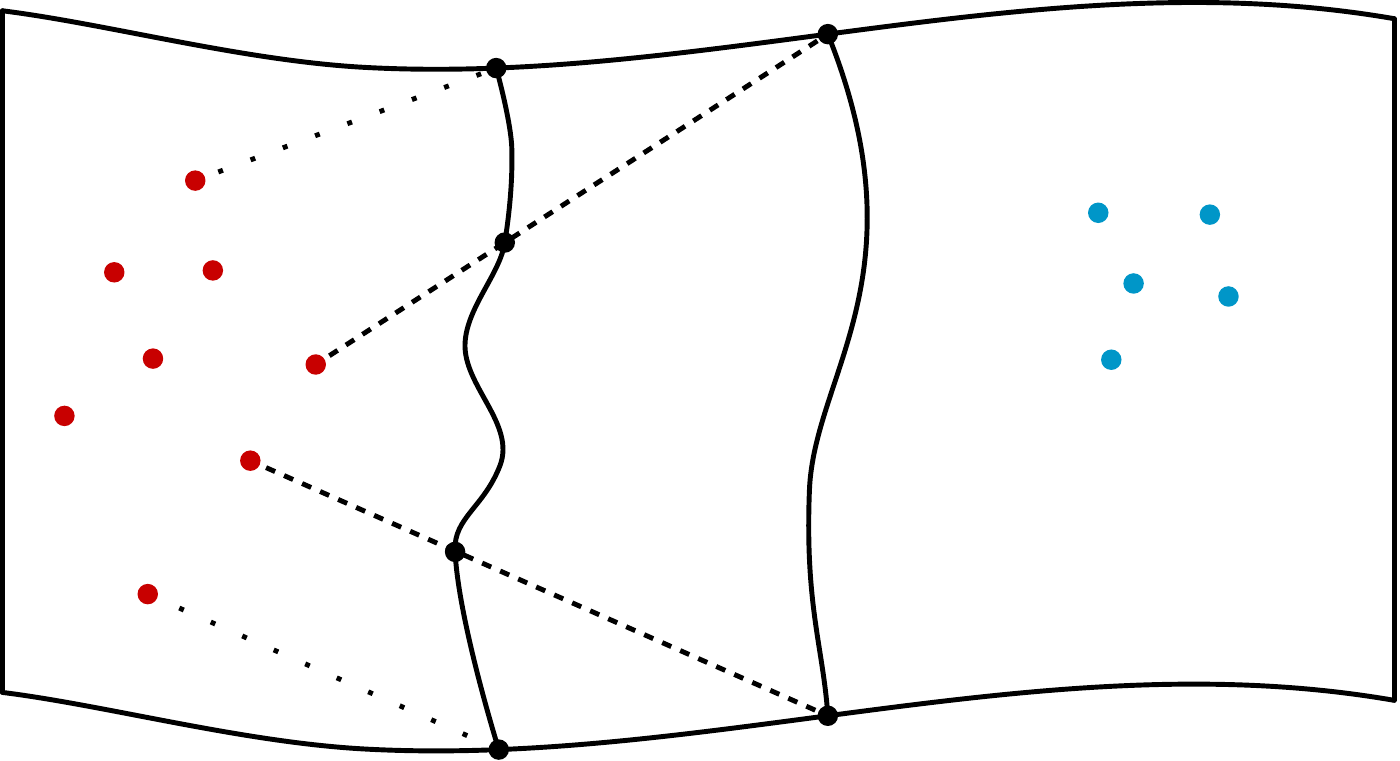}
\put(34,-3.8){$x_i$}
\put(57.5,-1.5){$x_j$}
\put(35,52.5){$y_i$}
\put(58,55){$y_j$}

\put(27.5,25){$p_i$}
\put(61,26.5){$p_j$}

\put(6,39){$L$}
\put(89,38){$R$}

\put(18,4){$\alpha$}
\put(20,45){$\epsilon$}
\put(50,41){$\nu$}
\put(49,10){$\mu$}

\put(34,15){$v$}
\put(37.5,34){$u$}

\end{overpic}
\end{subfigure}		
\qquad
	\begin{subfigure}{5cm}
\begin{overpic}[width=5cm,percent]{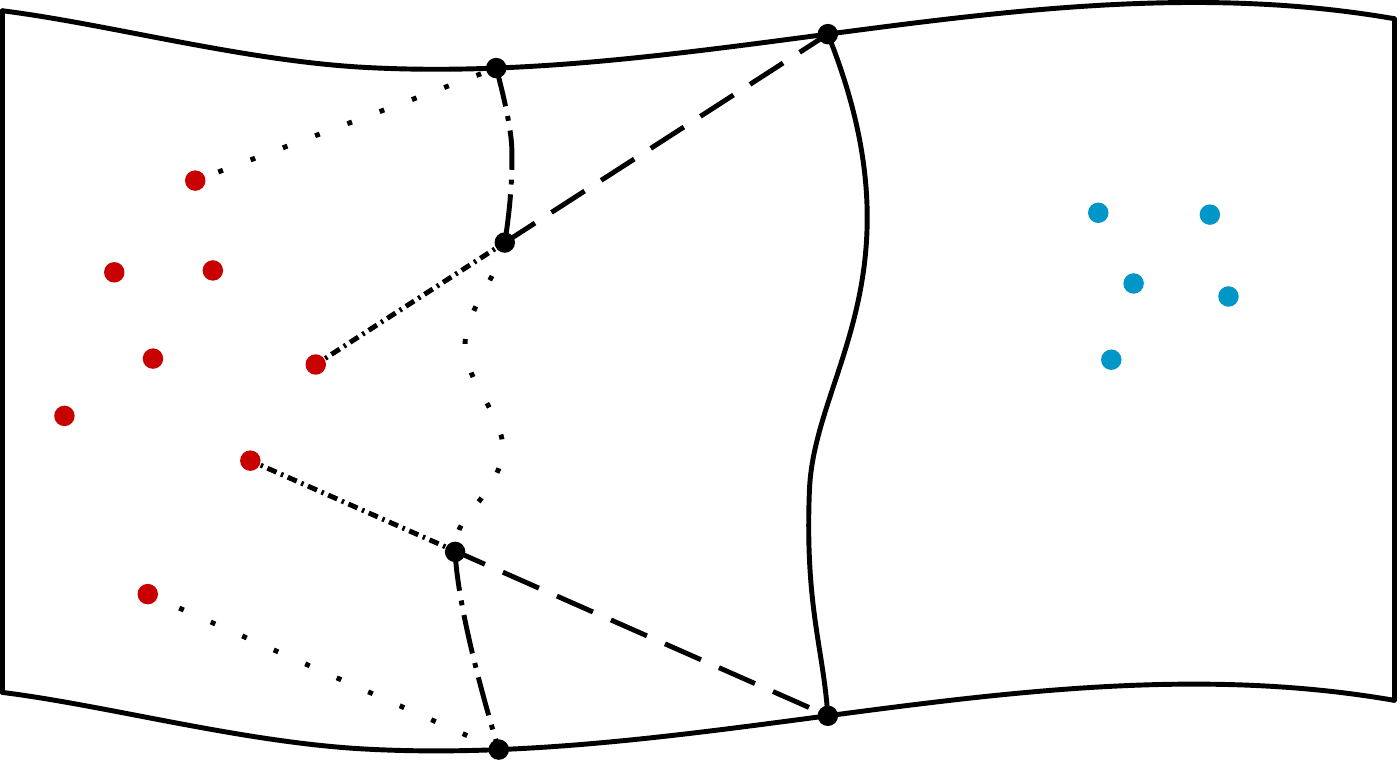}
\put(34,-3.8){$x_i$}
\put(57.5,-1.5){$x_j$}
\put(35,52.5){$y_i$}
\put(58,55){$y_j$}

\put(18,4){$\alpha$}
\put(20,45){$\epsilon$}

\put(36,6){$\beta$}
\put(25,20.3){$\gamma$}
\put(38,43){$\zeta$}
\put(26,35){$\eta$}
\put(36,26){$\theta$}
\put(50,41){$\iota$}
\put(49,9.5){$\kappa$}
\end{overpic}
\end{subfigure}	

\caption{example of paths and subpaths used in the proof of~\ref{claim:left_right}.}
\label{fig:alfabeto_greco}
\end{figure}

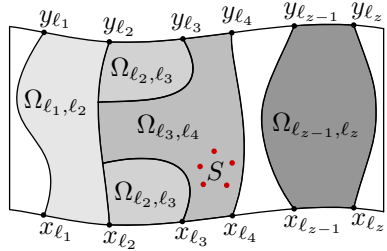
\begin{figure}[h]
\captionsetup[subfigure]{justification=centering}
\centering
\begin{subfigure}{5cm}
\begin{overpic}[width=5cm,percent]{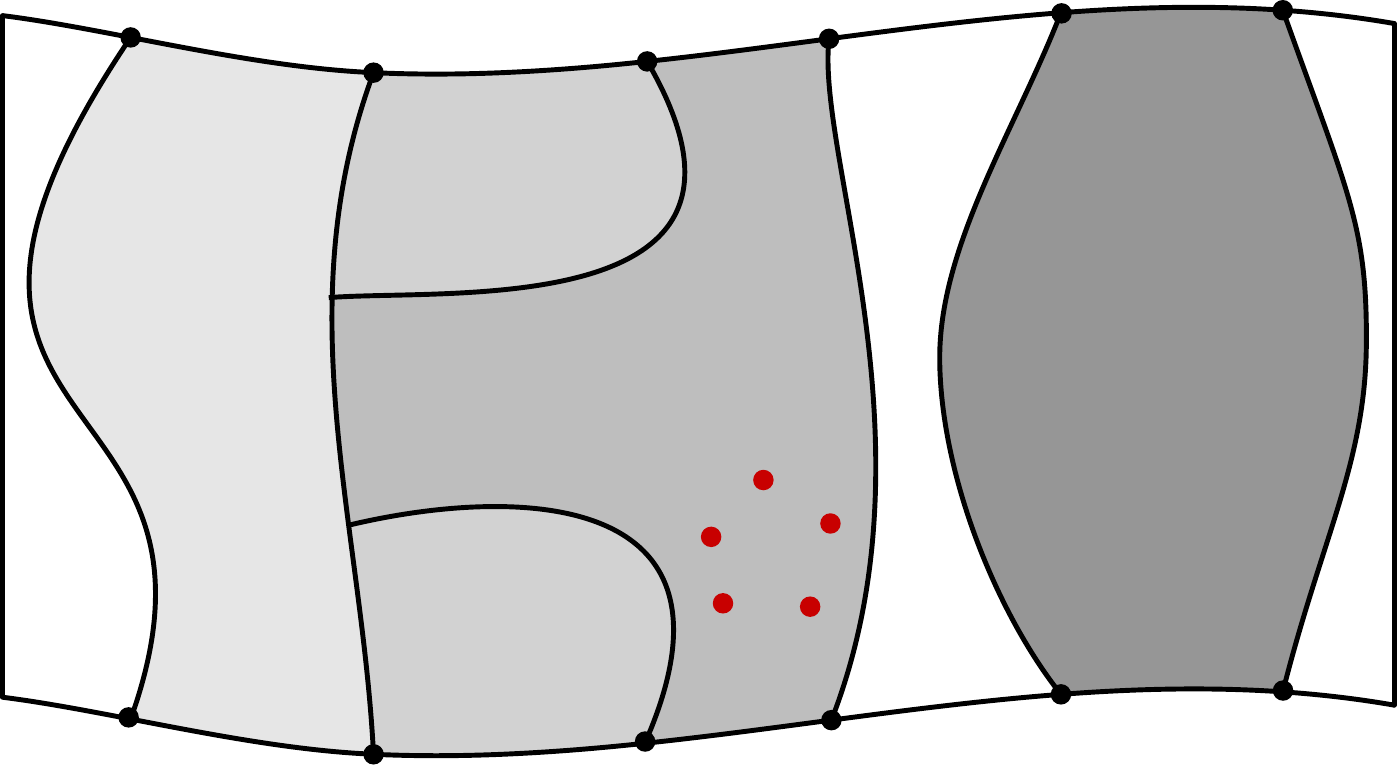}
\put(8,-1.5){$x_{\ell_1}$}
\put(8,55){$y_{\ell_1}$}
\put(25,-3.8){$x_{\ell_2}$}
\put(25,53){$y_{\ell_2}$}
\put(44,-2.9){$x_{\ell_3}$}
\put(43,54){$y_{\ell_3}$}
\put(57,-1.8){$x_{\ell_4}$}
\put(56.5,56){$y_{\ell_4}$}

\put(73,0.5){$x_{\ell_{z-1}}$}
\put(73,57){$y_{\ell_{z-1}}$}
\put(91,1){$x_{\ell_z}$}
\put(91,57){$y_{\ell_z}$}

\put(4,32){$\Omega_{\ell_1,\ell_2}$}
\put(28,7.5){$\Omega_{\ell_2,\ell_3}$}
\put(27,40.5){$\Omega_{\ell_2,\ell_3}$}
\put(34,25.5){$\Omega_{\ell_3,\ell_4}$}
\put(70,26){$\Omega_{\ell_{z-1},\ell_z}$}

\put(52.5,13){$S$}

\end{overpic}
\end{subfigure}
\caption{by Lemma~\ref{lemma:Omega_i}, it holds that $\min_{\ell\in L_r}d_\ell(S)\geq\min\{{d_{\ell_{3}}(S),d_{\ell_{4}}(S)}\}-\delta$.}
\label{fig:application_alfabeto_greco}
\end{figure}

To compute distances in $D$ we have to solve some SSSP instances in some $\Omega_{i,j}$'s subsets. These subsets can share boundary edges, thus the sum of their edges might be $O(n^2)$. We note that, by the single-touch property, if an edge $e$ belongs to $\Omega_{i,j}$ and $\Omega_{j,\ell}$ for some $i<j<\ell\in [k]$, then $e\in p_j$.

To overcome this problem we introduce subsets $\widetilde{\Omega}_{i,j}$ in the following way: for any $i<j\in[k]$, if $p_i\cap p_j$ is a non-empty path $q$, then we define $\widetilde{\Omega}_{i,j}$ as $\Omega_{i,j}$ in which we replace path $q$ by an edge with the same length; note that the single-touch property implies that all vertices in $q$ but its extremal have degree two. Otherwise, we define $\widetilde{\Omega}_{i,j}=\Omega_{i,j}$. Note that distances between vertices in $\widetilde{\Omega}_{i,j}$ are the same as in $\Omega_{i,j}$. It the following lemma we show how to compute some $\widetilde{\Omega}_{i,j}$'s in $O(n)$ time. 

\begin{lemma}\label{lemma:Henzinger_in_Omega_i}
Let $A=(a_1,a_2,\ldots,a_z)$ be any increasing sequence of indices of $[k]$. It holds that \\$\sum_{i\in[z-1]}|E(\widetilde{\Omega}_{a_i,a_{i+1}})|=O(n)$. Moreover, given $U$, we compute $\widetilde{\Omega}_{a_i,a_{i+1}}$, for all $i\in[z-1]$, in $O(n)$ total time.
\end{lemma}
\begin{proof}
For convenience, we denote by ${\Omega}_i$ the set ${\Omega}_{a_i,a_{i+1}}$, for all $i\in[z-1]$. We note that if $e\in {\Omega}_i\cap {\Omega}_{i+1}$, then $e\in p_{i+1}$. Thus, if $e$ belongs to more than two $\Omega_i$'s, then $e$ belongs to exactly two $\widetilde{\Omega}$'s because it is contracted in all other $\Omega_i$'s by definition of the $\widetilde{\Omega}_i$'s. Thus $\sum_{i\in[z-1]}|E(\widetilde{\Omega}_i)|=O(n)+O(z)=O(n)$ because $z\leq k\leq n$.

To obtain all $\widetilde{\Omega}_i$'s, we compute $U_z=\bigcup_{a\in A}p_a$ in $O(n)$ time by Theorem~\ref{th:U_italiano+err_giappo}. Then we preprocess all trees in $U_z$ in $O(n)$ time by using Gabow and Tarjan's result~\cite{gabow-tarjan} in order to obtain the intersection path $p_{a_i}\cap p_{a_{i+1}}$ via lowest common ancestor queries, and its length in $O(1)$ time with a similar approach. Finally, we build $\widetilde{\Omega}_i$ in $O(|E(\widetilde{\Omega}_i)|)$, for all $i\in[z-1]$, with a BFS visit of $\Omega_i$ that excludes vertices of $p_{a_i}\cap p_{a_{i+1}}$.\qed
\end{proof}

\section{Computing edge vitality}\label{sec:complexity_edge}


Now we can give our main result about edge vitality stated in Theorem~\ref{th:main_real_edge}. We need the following preliminary lemma that is an easy consequence of Lemma~\ref{lemma:Omega_i} and Lemma~\ref{lemma:Henzinger_in_Omega_i}.

\begin{lemma}\label{lemma:cost_L_r}
Let $r\in\mathbb{N}$, given $U$, we compute a value $\alpha_r(e)\in[\min_{i\in L_r}\{d_i(e)\}$, $\min_{i\in L_r}\{d_i(e)\}+\delta)$ for all  $e\in E(D)$ in $O(n)$ time.
\end{lemma}
\begin{proof} 
We compute $U_{r}=\bigcup_{i\in L_r}p_i$ in $O(n)$ time by Theorem~\ref{th:U_italiano+err_giappo}. Let $e\in E(D)$. If $e\in U_r$, we set $\alpha_r(e)=\MF+\delta(r+1)-w(e)$.  If $e\not\in U_r$, then either $e\in\widetilde{\Omega}_{\ell_i,\ell_{i+1}}$ for some $i\in[z_r-1]$, or $e\in\Left_{\ell_1}$, or $e\in\Right_{\ell_{z}}$.

For all $e\subseteq\Left_{\ell_1}$, we set $\alpha_r(e)=d_{\ell_1}(e)$, similarly, for all $e\subseteq\Right_{\ell_z}$, we set $\alpha_r(e)=d_{\ell_z}(e)$. Finally, if $e\in \widetilde{\Omega}_{\ell_i,\ell_{i+1}}$, then we set $\alpha_r(e)=\min\{d_{\ell_i}(e),d_{\ell_{i+1}}(e)\}$. All these choices satisfy the required estimation by Lemma~\ref{lemma:Omega_i}.

To compute required distances, it suffices to solve two SSSP instances with sources $x_i$ and $y_i$ to vertices of $\widetilde{\Omega}_i\cup \widetilde{\Omega}_{i+1}$, for every $i\in L_r$. In total we spend $O(n)$ time by Lemma~\ref{lemma:Henzinger_in_Omega_i} by using algorithm in~\cite{henzinger} for SSSP instances.\qed
\end{proof}

{
\renewcommand{\thetheorem}{\ref{th:main_real_edge}}
\begin{theorem}
Let $G$ be a planar graph with positive edge capacities. Then for any $c,\delta>0$, we compute a value $vit^\delta(e)\in(vit(e)-\delta,vit(e)]$ for all $e\in E(G)$ satisfying $c(e)\leq c$, in $O(\frac{c}{\delta}n+n\log\log n)$  time.
\end{theorem}
\addtocounter{theorem}{-1}
}
\begin{proof}
We compute $U$ in $O(n\log\log n)$ time by Theorem~\ref{th:U_italiano+err_giappo}. If $d_i>\MF+c(e)$, then $d_i(e^D)>\MF$, so we are only interested in computing (approximate) values of $d_i(e^D)$ for all $i\in[k]$ satisfying $d_i<\MF+c$. By Lemma~\ref{lemma:cost_L_r}, for each $r\in\{0,1,\ldots,\lceil\frac{c}{\delta}\rceil\}$, we compute $\alpha_r(e^D)\in[\min_{i\in L_r}d_i(e^D),$ $\min_{i\in L_r}d_i(e^D)+\delta)$, for all $e^D\in E({D})$, in $O(n)$ time. Then, for each $e^D\in E({D})$, we compute $\alpha(e^D)=\min_{r\in\{0,1,\ldots,\frac{c}{\delta}\}}\alpha_r(e^D)$; it holds that $\alpha(e^D)\in[\min_{i\in[k]}\{d_i(e^D)\}$, $\min_{i\in[k]}\{d_i(e^D)\}+\delta)$. Then, by Proposition~\ref{prop:vitality_edges}, for each $e\in E(G)$ satisfying $c(e)\leq c$, we compute a value $vit^\delta(e)\in(vit(e)-\delta,vit(e)]$ in $O(1)$ time.\qed
\end{proof}

\section{Computing vertex vitality}\label{sec:complexity_vertex}

In this section we show how to compute vertex vitality by computing an additive guaranteed approximation of distances required in Proposition~\ref{prop:vitality_vertices}.

Let us denote by $F$ the set of faces of $D$. By Proposition~\ref{prop:vitality_vertices}, for every face $f\in F$ we need $\min_{i\in[k]}\{d_i(f)\}$, this is  discussed in Lemma~\ref{lemma:cost_L_r_face}. For faces $f\in F^x=\{f\in F \, |\,$ $f$ and $\pi_x$ have common vertices$\}$ we need also $\min_{i\in[k]}\{d_i(q^y_f)\}$ and $\dist_{D}(f_v,q^y_{f})$. Similarly, for faces $f\in F^y=\{f\in F \, |\,$ $f$ and $\pi_y$ have common vertices$\}$ we need also $\min_{i\in[k]}\{d_i(q^x_f)\}$ and $\dist_{D}(f_v,q^x_{f})$.

We observe that there is symmetry between $q^x_f$ and $q^y_f$. Thus we restrict some definitions and results to the ``$y$ case'' and then we use the same results for the ``$x$ case''. In this way, we have to show only how to compute $\dist_{D}(f,q^y_f)$ (it is done in Subsection~\ref{subsection:dist_D(f,q^y_f)}) and $\min_{i\in[k]}\{d_i(q^y_{f})\}$ (see Subsection~\ref{subsection:d_i(q^y_f)}) for every face $f\in F$ that intersects $\pi_y$ on vertices.

By using the same procedure of Lemma~\ref{lemma:cost_L_r}, we can also computing $d_i(f)$ for $f\in F$. Thus we can state the following lemma.  

\begin{lemma}\label{lemma:cost_L_r_face}
Let $r\in\mathbb{N}$, given $U$, we compute a value $\alpha_r(f)\in[\min_{i\in L_r}\{d_i(f)\},$ $\min_{i\in L_r}\{d_i(f)\}+\delta)$ for all $f\in  F$ in $O(n)$  time.
\end{lemma}

\subsection{Computing $\dist_D(f,q^y_f)$}\label{subsection:dist_D(f,q^y_f)}

The unique our result of this subsection is stated in Lemma~\ref{lemma:centro_faccia_nlogn}. To obtain it, we use the following result that easily derives from Klein's algorithm about the multiple source shortest path problem~\cite{klein2005}.

\begin{theorem}[\cite{klein2005}]\label{th:klein}
Given an $n$ vertices undirected planar graph $G$ with nonnegative edge-lengths, given $r$ pairs $\{(a_i,b_i)\}_{i\in[r]}$ where the $b_i$'s are on the boundary of the infinite face and the $a_i$'s are anywhere, it is possible to compute $dist_G(a_i,b_i)$, for all $i\in[r]$, in $O(r\log n+n\log n)$ time and $O(n)$ space.
\end{theorem}

\begin{lemma}\label{lemma:centro_faccia_nlogn}
We compute $\dist_{D}(f,q^y_f)$, for all $f\in F^y$, in $O(n\log n)$  time.
\end{lemma}
\begin{proof}
For every $i\in[k]$ let $F_i\subseteq F^y$ be the set of faces such that $x_i\in f$, for all $f\in F_i$. We observe that if $|F_i|=m$, then $deg_D(x_i)\geq m+1$, where $deg_D(x_i)$ is the degree of $x_i$ in $D$.

Let $D'$ be the graph obtained by adding a new vertex  $u_f$ for each face $f\in F^y$ and connecting  $u_f$ to all vertices of $f$ by an edge of length $L$, where $L=\sum_{e\in E({D})}w(e)$ (see Figure~\ref{fig:face_nlogn} for an example of construction of graph $D'$). Thus $\dist_{D}(y_i,f)=\dist_{D'}(y_i,u_f)-L$.

We compute $d_{D'}(y_i,u_f)$, for $i\in[k]$ and  $f\in F_i$, by using the result stated in Theorem~\ref{th:klein}. Being $|V(D')|=O(n)$, we spend $O\big(\log n \sum_{i\in[k]}|F_i|+n\log n\big)\leq O\big(\log n\sum_{i\in[k]}(deg_{{D}}(x_i)-1)+O(n\log n)\big)=O(n\log n+n\log n)=O(n\log n)$ time. Finally, for all $f\in F^y$
$$\dist_{D}(f,q^y_f)=\min_{\{i\in[k]\,|\,x_i\in f\}}\dist_{D}(y_i,f)=\min_{\{i\in[k]\,|\,x_i\in f\}}\dist_{D'}(y_i,u_f)-L.$$ 
Thus we obtain what we need adding no more time than $\sum_{\{i\in[k]\,|\,x_i\in f\}}O(1)\leq O(\sum_{f\in F^y}|V(f)|)\leq O(\sum_{f\in  F}|V(f)|)=O(n)$.\qed
\end{proof}

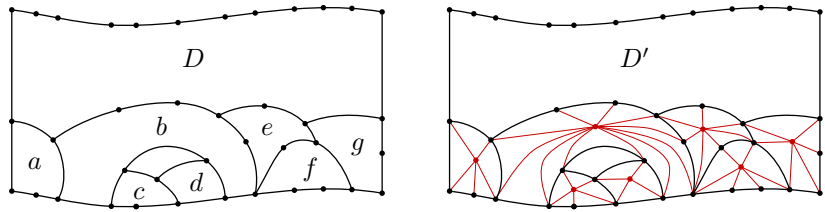
\begin{figure}[h]
\captionsetup[subfigure]{justification=centering}
\centering
\begin{subfigure}{5cm}
\begin{overpic}[width=5cm,percent]{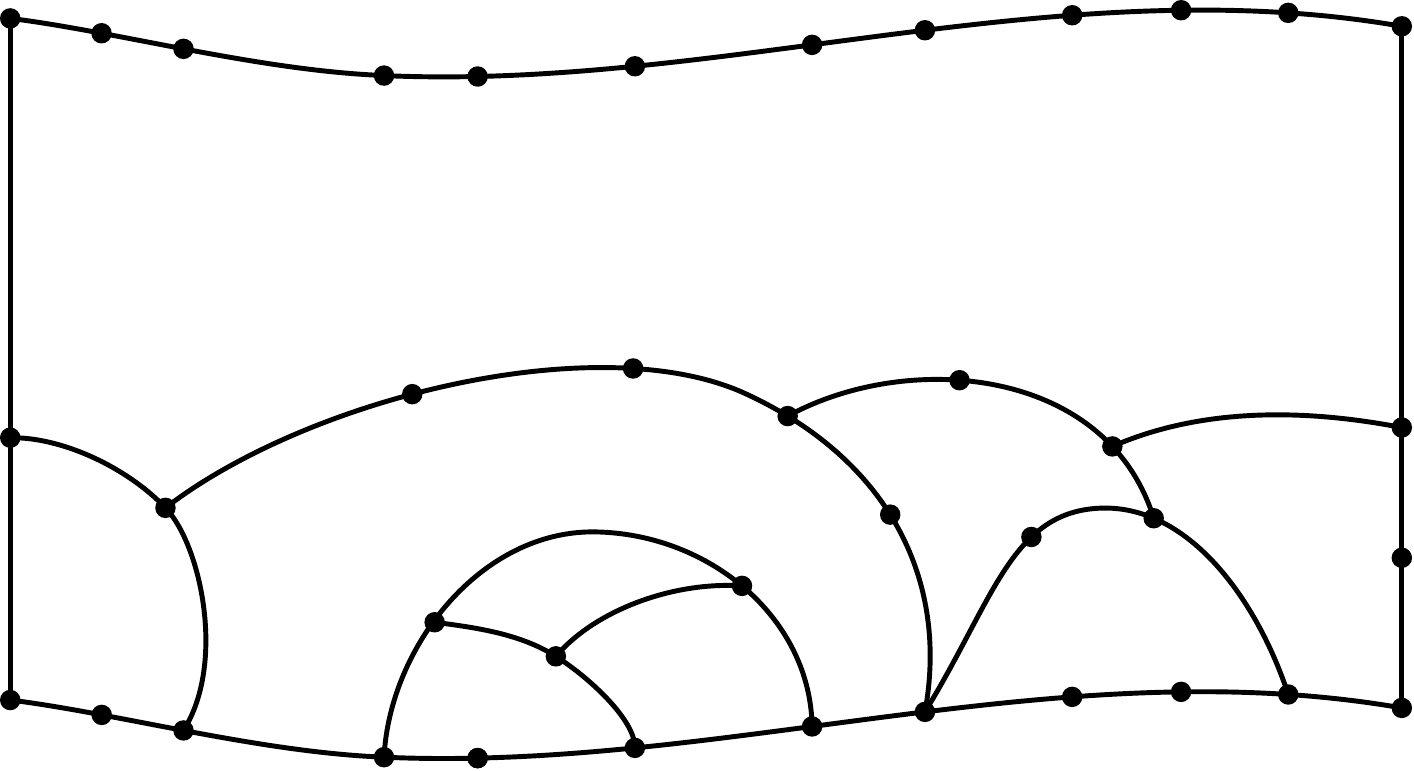}
\put(5,11){$a$}
\put(39,20){$b$}
\put(33,3){$c$}
\put(48,5){$d$}
\put(67,20){$e$}
\put(78,9){$f$}
\put(91,16){$g$}
\put(46,38){$D$}

\end{overpic}
\end{subfigure}		
\qquad
\begin{subfigure}{5cm}
\begin{overpic}[width=5cm,percent]{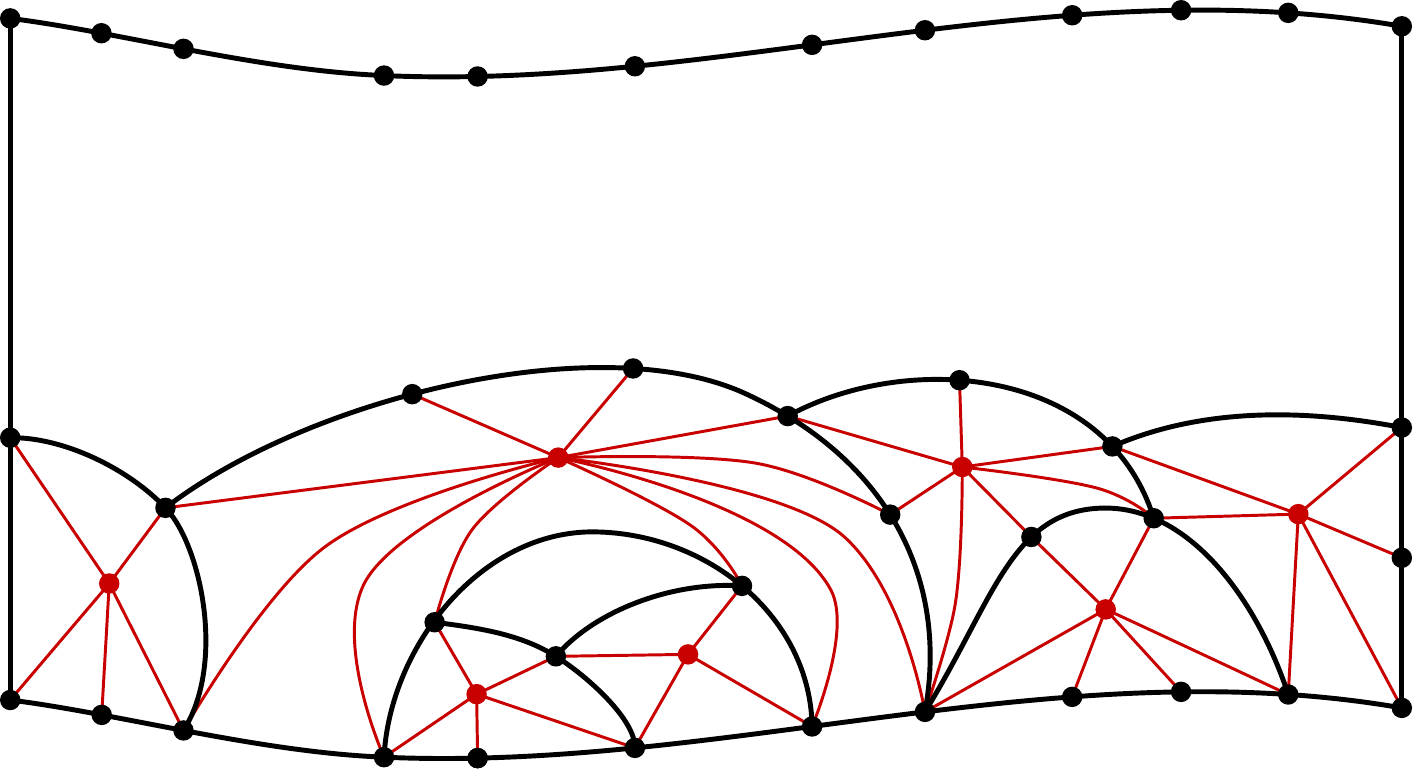}
\put(46,38){$D'$}
\end{overpic}
\end{subfigure}		
 \caption{graph $D$, faces in $F^y$ and graph $D'$ used in the proof of Lemma~\ref{lemma:centro_faccia_nlogn}.}
\label{fig:face_nlogn}
\end{figure}

\subsection{Computing $d_i(q^y_f)$}\label{subsection:d_i(q^y_f)}

We note that for computing the $d_i(q^y_f)$'s we can not directly use  Lemma~\ref{lemma:Omega_i} as we have done for the $d_i(e)$'s and the $d_i(f)$'s.  Indeed, it is possible that vertices in $q^y_f$ are not contained in any slice $\Omega_{i,j}$, with $i,j$ consecutive indices in $L_r$. To overcome this, we have to introduce a partial order on faces of $D$.

For all $f\in F^y$, we define $f^-$ and $f^+$ as the minimum and maximum indices in $[k]$, respectively, such that $x_{f^-},x_{f^+}\in V(f)$. Now we introduce the concept of \emph{maximal face}. Let $f\in F^y$ and let $p_f$ and $q_f$ be the two subpaths of the border cycle of $f$ from $x_{f^-}$ to $x_{f^+}$. We say that $g\prec f$ if $g$ is contained in the region $R$ bounded by $\pi_x[x_{f^-},x_{f^+}]\circ p_f$, this implies that $g$ is also contained in the region $R'$ bounded by $\pi_x[x_{f^-},x_{f^+}]\circ q_f$, thus the definition does not depend on the choice of $p_f$ and $q_f$. Finally, we say that $f$ is \emph{maximal} if it does not exist any face $g\in F^y$ satisfying $f\prec g$, and we define $F_{max}=\{f\in F^y \,|\, f \text{ is maximal}\}$, see the left part of  Figure~\ref{fig:maximal_faces}. We find $F_{max}$ in $O(n)$  time.

Given $r\in\mathbb{N}$ and $f\in F^y$, we define $f^+_r$ as the smallest index in $L_r$ such that $f^+< f^+_r$ (if $f^+>\ell^r_{z_r}$, then we define $f^+_r=\ell^r_{z_r}$). Similarly, we define $f^-_r$ as the largest index in $L_r$ such that $f^-_r> f^-$ (if $f^-<\ell^r_1$, then we define $f^-_r=\ell^r_1$), see the right part of Figure~\ref{fig:maximal_faces}.

\begin{figure}[h]
\captionsetup[subfigure]{justification=centering}
\centering
	\begin{subfigure}{5cm}
\begin{overpic}[width=5cm,percent]{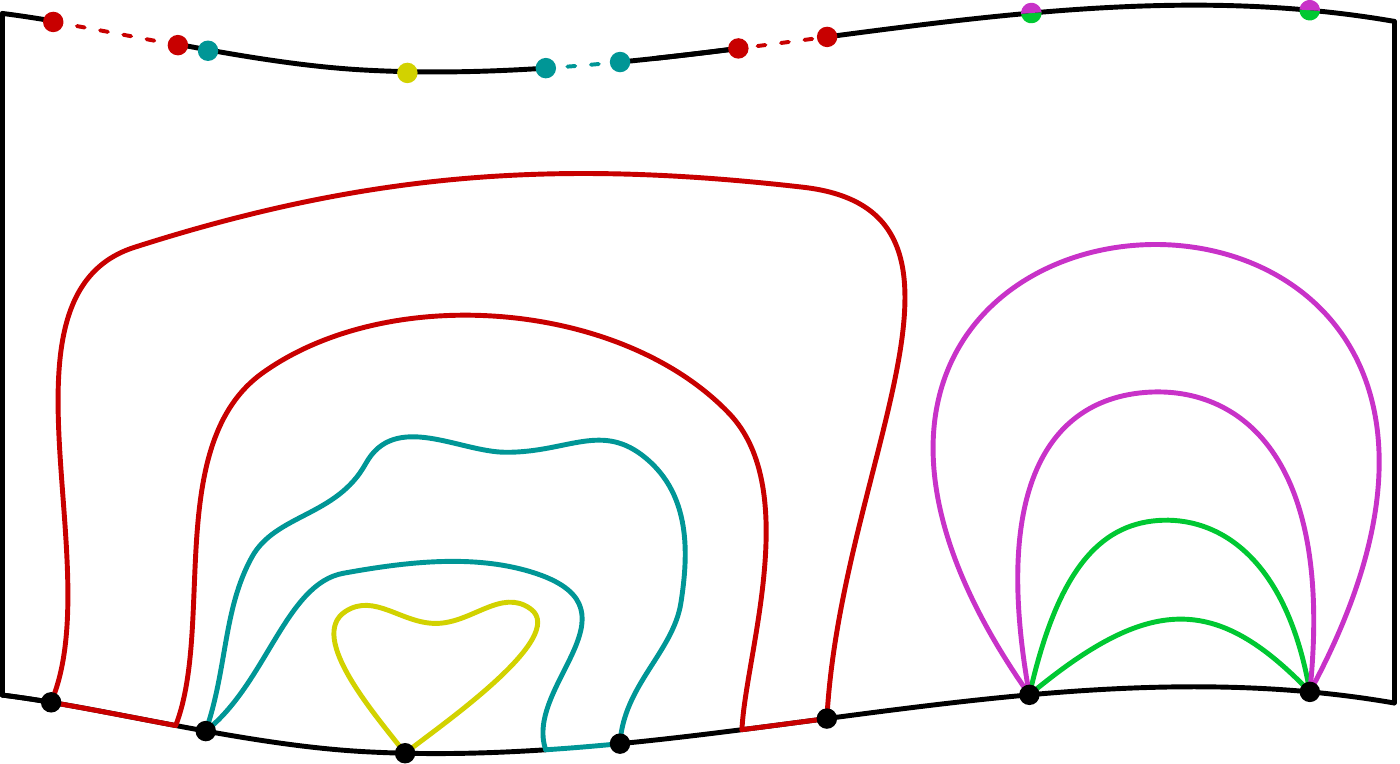}
\put(33,35){$a$}
\put(36,16){$b$}
\put(29,5){$c$}
\put(80,30){$d$}
\put(82,12.5){$e$}

\put(1,-0.5){$x_{a^-}$}
\put(58,-1.5){$x_{a^+}$}

\put(13,-2.7){$x_{b^-}$}
\put(43,-3.2){$x_{b^+}$}

\put(27,-4){$x_\alpha$}

\put(73,0){$x_\beta$}
\put(92,0.5){$x_\gamma$}
\end{overpic}
\end{subfigure}
\qquad
	\begin{subfigure}{5cm}
\begin{overpic}[width=5cm,percent]{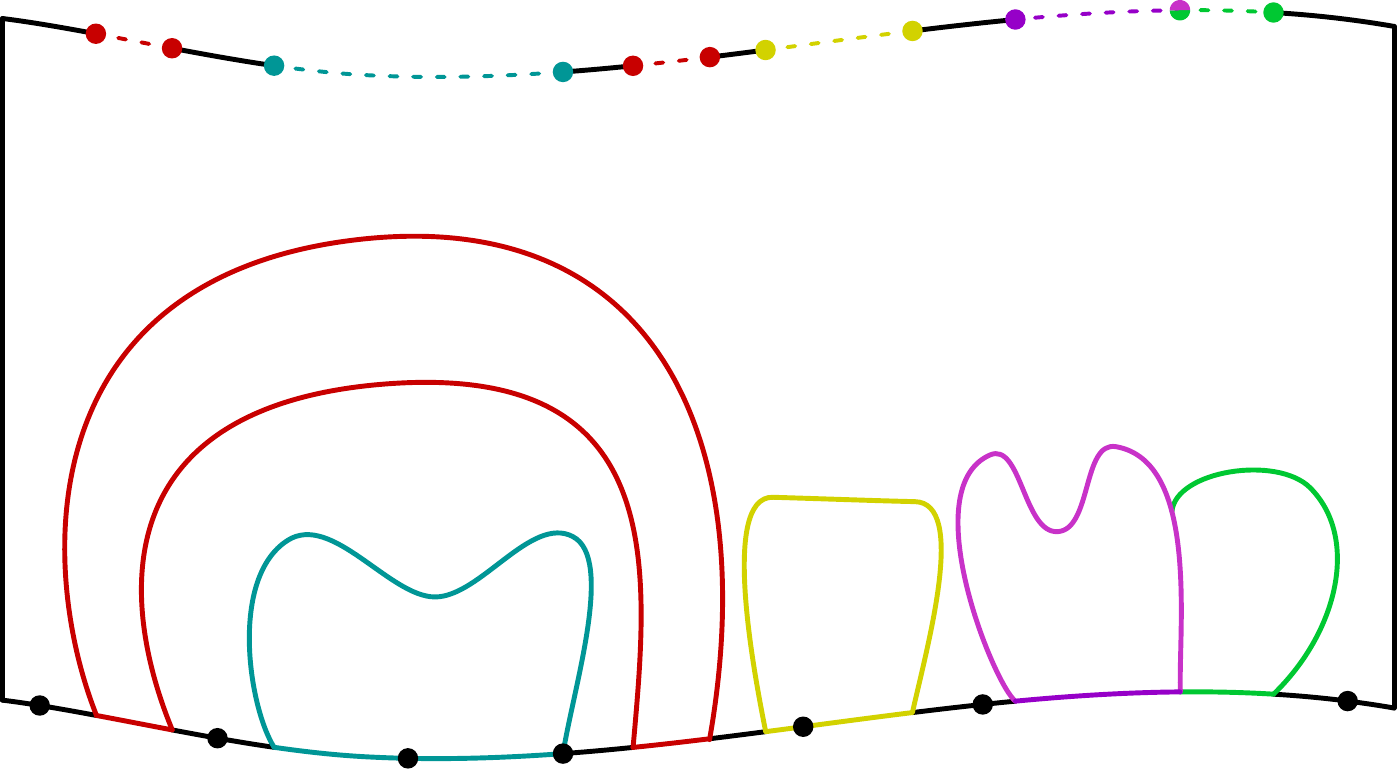}
\put(28,31){$f$}
\put(30,5){$g$}
\put(59,10){$h$}
\put(75,10){$i$}
\put(87,12){$j$}

\put(1,-0.5){$x_{\ell_1}$}
\put(13.7,-2.6){$x_{\ell_2}$}
\put(27,-3.8){$x_{\ell_3}$}
\put(38.5,-3.5){$x_{\ell_4}$}
\put(55.5,-1.8){$x_{\ell_5}$}
\put(68.5,0){$x_{\ell_6}$}
\put(94.5,0.4){$x_{\ell_7}$}
\end{overpic}
\end{subfigure}
  \caption{on the left $c\prec b\prec a$ and $e\prec d$; $a$ and $d$ are the only maximal faces; it holds that $c^-=c^+=\alpha$, $[d^-,d^+]=[e^-,e^+]=[\beta,\gamma]$. On the right let $L_r=(\ell_1,\ell_2,\ldots,\ell_7)$, it holds that: $[f^-_r,f^+_r]=[\ell_1,\ell_5]$, $[g^-_r,g^+_r]=[\ell_2,\ell_5]$, $[h^-_r,h^+_r]=[\ell_4,\ell_6]$, $[i^-_r,i^+_r]=[j^-_r,j^+_r]=[\ell_6,\ell_7]$.}
\label{fig:maximal_faces}
\end{figure}
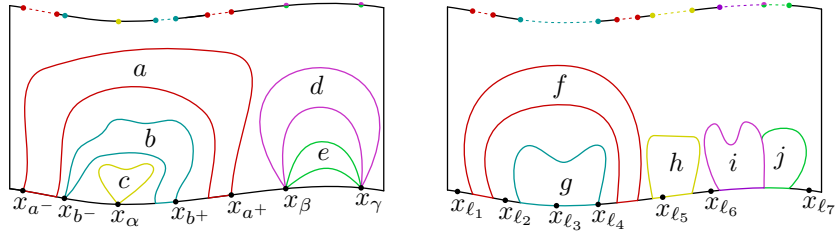

Now we deal with computing $d_i(q^y_f)$, for all $f\in F^y$.  By following Equation \eqref{eq:vitality_vertices}, we can restrict only to the easier case in which $f$ satisfies $d_i(q^y_f)<\dist_{D}(f,q^y_f)$; indeed, if $f$ does not satisfy it, then we are not interested in the value of $d_i(q^y_f)$.

\begin{lemma}\label{lemma:cost_L_r_vertex}
Let $r\in\mathbb{N}$. Given $\dist_{D}(f,q^y_f)$ and given $U$, for all $f\in F^y$ satisfying $\min_{i\in L_r}d_i(q^y_f)<\dist_{D}(f,q^y_f)$ we compute a value $\beta_r(f)\in[\min_{i\in L_r}d_i(q^y_f),$ $ \min_{i\in L_r}d_i(q^y_f)+\delta)$ in $O(n)$  total time.
\end{lemma}
\begin{proof}
Let $f\in F^y$. We observe that if $i\in[f^-,f^+]$, then every path from $x_i$ to $q^y_f$ passes through either $x_{f^-}$ or $x_{f^+}$. Thus, for every $i\in[f^-,f^+]$, it holds that $d_i(q^y_f)\geq \dist_{D}(f,q^y_f)$. Hence for any $f\in F^y$ satisfying $\min_{i\in L_r}d_i(q^y_f)<\dist_{D}(f,q^y_f)$ it holds that 
\begin{equation}\label{eq:min_1}
\min_{i\in L_r}d_i(q^y_f)=\min_{i\in L_r,i\not\in[f^-,f^+]}\{d_i(q^y_f)\},
\end{equation}
%
being $q^y_f\subseteq\Right_{f^-_r}$ and $q^y_f\subseteq\Left_{f^+_r}$, then Lemma~\ref{lemma:Omega_i} and Equation \eqref{eq:min_1} imply
\begin{equation}\label{eq:min_2}
\min_{i\in L_r}d_i(q^y_f)=\min_{i\in L_r,i\not\in[f^-,f^+]}\{d_i(q^y_f)\}\geq\min\{d_{f^-_r}(q^y_f),d_{f^+_r}(q^y_f)\}-\delta.
\end{equation}

To complete the proof, we need to show how to compute $d_{f^-_r}(q^y_f)$ and $d_{f^+_r}(q^y_f)$, for all $f\in F^y$ satisfying $\min_{i\in L_r}d_i(q^y_f)<\dist_{D}(f,q^y_f)$ in $O(n)$  time. In the following claim we prove it by removing the request that every face $f\in F^y$ has to satisfy $\min_{i\in L_r}d_i(q^y_f)<\dist_{D}(f,q^y_f)$.

\begin{enumerate}[label=\alph{CONT})]\stepcounter{CONT}
\item\label{claim:long} We compute $d_{f^-_r}(q^y_f)$ and $d_{f^+_r}(q^y_f)$, for all $f\in F^y$, in $O(n)$  time.
\end{enumerate}
\claimbegin{\ref{claim:long}} we recall that $d_i(q^y_f)=\dist_{D}(x_i,q^y_f)+\dist_{D}(y_i,q^y_f)$, for all $i\in[k]$ and $f\in F^y$. Being $q_f^y\subseteq V(\pi_y)$ we compute $\dist_{D}(y_i,q^y_f)$ in $O(|V(q^y_f)|)$ time. Thus we have to compute only $\dist_{D}(x_i,q^y_f)$, for required $i\in L_r$ and $f\in F^y$.

For every $f\in F^y$, let $R_f=\Omega_{f^-_r,f^+_r}$, and let $\mathcal{R}=\bigcup_{f\in F_{max}}R_f$. We observe that, given two maximal faces $f$ and $g$, it is possible that $R_f=R_g$. This happens if and only if $f^-_r=g^-_r$ and $f^+_r=g^+_r$ (see face $i$ and face $j$ in Figure~\ref{fig:f-}). We overcome this abundance by introducing $\widetilde{F}$ as a minimal set of faces such that $\mathcal{R}=\bigcup_{f\in\widetilde{F}}R_f$ and $R_f\neq R_g$, for all distinct $f,g\in\widetilde{F}$ (see Figure~\ref{fig:f-} for an example of $\widetilde{F}$).

For every $f\in\widetilde{F}$, it holds that $\pi_y[{f}^-_r,{f}^+_r]\subseteq R_{f}$. Thus, by the above argument, if $g\in F^y$ and $R_g\subseteq R_{f}$, then $q^y_g\subseteq R_{f}$.
%
We solve 4 SSSP instances in $R_f$ with sources $x_j$, for all $j\in\{f^-_r,f^+_r,f^-,f^+\}$ (possibly, ${f}^-_r={f}^-$ and/or ${f}^+_r={f}^+$ and/or $f^-=f^+$). Now we have to prove that this suffices to compute $d_{f^-_r}(q^y_f)$ and $d_{f^+_r}(q^y_f)$, for all $f\in F^y$. In particular we show that, after solving the SSSP instances, we compute $d_{g^-_r}(q^y_g)$ and $d_{g^+_r}(q^y_g)$ in $O(|V(g)|)$ time, for each $g\in F^y$.

Let $g\in F^y$, and let $f\in\widetilde{F}$ be such that $g\subseteq R_{f}$. There are two cases: either $g^-_r={f}^-_r$ and $g^+_r={f}^+_r$, or $g^-_r\neq{f}^-_r$ and/or $g^+_r\neq{f}^+_r$.

If the first case occurs, then we compute $\dist_{D}(x_{g^-_r},q^y_g)=\dist_{D}(x_{f^-_r},q^y_g)$ and $\dist_{D}(x_{g^+_r},q^y_g)=\dist_{D}(x_{f^+_r},q^y_g)$ in $O(|V(g)|)$ time, because $q^y_g\subseteq R_f$ and $|V(q^y_g)|<|V(g)|$. Otherwise, w.l.o.g., we assume that $g^-_r\neq f^-_r$ (if $g^+_r\neq f^+_r$, then the proof is similar). By definitions of $\widetilde{F}$, $\Omega_g$, and $\Omega_f$, it holds that $g\prec f$. Thus $g^-_r\in[f^-,f^+]$, therefore every path from $x_{g^-_r}$ to $q^y_g$ passes through either $f^-$ or $f^+$ (see $g_3$ and $f_5$ in Figure~\ref{fig:f-}). By this discussion, it follows that 
\begin{align*}
\dist_{D}(x_{g^-_r},q^y_g)&=\min\left\{\begin{array}{l}
\dist_{D}(x_{g^-_r},x_{f^-})+\dist_{D}(x_{f^-},q^y_f)\\
\dist_{D}(x_{g^-_r},x_{f^+})+\dist_{D}(x_{f^+},q^y_f)
\end{array}\right\}\\
&=\min\left\{\begin{array}{l}
|\pi[x_{g^-_r},x_{f^-}]|+\dist_{D}(x_{f^-},q^y_f)\\
|\pi[x_{g^-_r},x_{f^+}]|+\dist_{D}(x_{f^+},q^y_f)
\end{array}\right\}
\end{align*}
%
%
We compute all these distances by the solutions of previous SSSP instances in $O(|V(q^y_g)|)$ time, and thus we compute $\dist_{D}(x_{g^-_r},q^y_g)$ in $O(|V(q^y_g)|)$ time. By symmetry, the same cost is required to compute $\dist_{D}(x_{g^+_r},q^y_g)$.

We have proved that, after solving the described SSSP instances, we compute $d_{f^-_r}(q^y_f)$ and $d_{f^+_r}(q^y_f)$, for all $f\in F^y$, in $O(|V(f)|)$ time for each $f\in F^y$. Being $\sum_{f\in F^y}|V(f)|=O(n)$, it remains to show that we can solve all the previous SSSP instances in $O(n)$  time. We want to use Lemma~\ref{lemma:Henzinger_in_Omega_i} (we recall that, for our purposes, distances in $\Omega_{f^-_r,f^+_r}$ are the same in $\widetilde{\Omega}_{f^-_r,f^+_r}$).

Let us fix $i\in[h]$ and let $a=f_i$, $b=f_{i+1}$, $c=f_{i+2}$ and $d=f_{i+3}$. We can not use directly Lemma~\ref{lemma:Henzinger_in_Omega_i} because it is possible that $a^-_r<b^+_r$ (see in Figure~\ref{fig:f-} $a=f_3$ and $b=f_4$, thus $b^-_r=\ell_4<\ell_5=a^+_r$) and thus we might have not an increasing set of indices. But, by definition of $\widetilde{F}$, it holds that $a_r^+\leq d_r^-$, indeed $a^+_r\in[b^-_r,c^+_r]$ otherwise $R_b=R_c$; these relations do not depend on $i$. Similarly, $d_r^-\geq a_r^+$. Thus we solve first the SSSP instances in $R_{f_i}$, for all $i\in[h]$ such that $i\equiv 0$ mod 3; then for $i\equiv 1$ mod 3 and finally for $i\equiv 2$ mod~3. By Lemma~\ref{lemma:Henzinger_in_Omega_i} it costs $O(n)$ time.
\claimend{\ref{claim:long}}

Finally, by Equation \eqref{eq:min_2}, we set $\beta_r(f)=\min\{d_{f^-_r}(q^y_f),d_{f^+_r}(q^y_f)\}$, for all $f\in F^y$ satisfying $\beta_r(f)<\dist_D(f,q^y_f)$ and we ignore faces in $F^y$ that do not satisfy it.\qed
\end{proof}

\begin{figure}[h]
\captionsetup[subfigure]{justification=centering}
\centering
	\begin{subfigure}{5cm}
\begin{overpic}[width=5cm,percent]{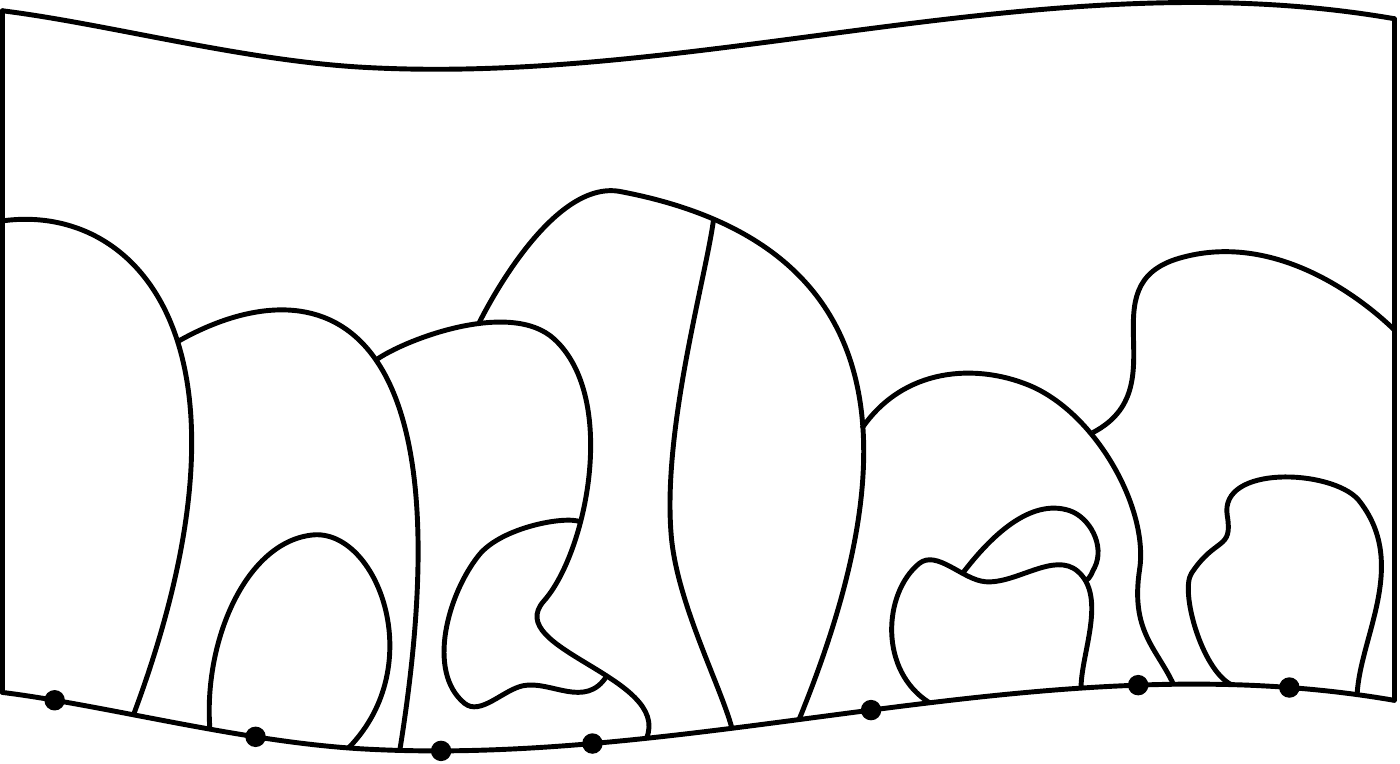}
\put(0,-0.5){$x_{\ell_1}$}
\put(16,-3.5){$x_{\ell_2}$}
\put(29.5,-4){$x_{\ell_3}$}
\put(41,-3.5){$x_{\ell_4}$}
\put(61,-1){$x_{\ell_5}$}
\put(79,0.5){$x_{\ell_6}$}
\put(90,0.5){$x_{\ell_7}$}

\put(4,17){$f_1$}
\put(18,21.5){$f_2$}
\put(32,22.5){$f_3$}
\put(50,18){$f_4$}
\put(65,19){$f_5$}
\put(87,24){$f_6$}

\put(19,8){$g_1$}
\put(41,33){$g_2$}
\put(67,8){$g_3$}
\put(89,11){$g_4$}
\end{overpic}
\end{subfigure}		
\caption{assume that $L_r=(\ell_1,\ldots,\ell_7)$. A possible $\widetilde{F}$ is $\widetilde{F}=\{f_1,\ldots,f_6\}$. Moreover, $g_1,g_3,g_4$ are not in $F_{max}$, $g_2\in F_{max}$ and $R_{g_2}=R_{f_4}$ thus $g_2\not\in\widetilde{F}$.}
\label{fig:f-}
\end{figure}
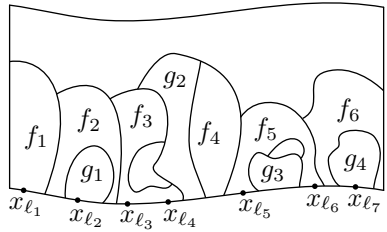

\subsection{Computational complexity of vertex vitality}

Now we give our theorems about vertex vitality. To prove Theorem~\ref{th:main_real_vertex} we follow the same approach used in Theorem~\ref{th:main_real_edge}, by referring to Proposition~\ref{prop:vitality_vertices} in place of Proposition~\ref{prop:vitality_edges}.

We recall that the result stated in Theorem~\ref{th:brutto} is more efficient than the result in Theorem~\ref{th:italiano_dynamic} if either $|S|<\log n$ and $E_S>|S|\log n$ or $|S|\geq\log n$ and $E_S>\frac{|S|n^{1/3}}{\log^{8/3}}$, where $E_S=\sum_{v\in S}deg(v)$.

{
\renewcommand{\thetheorem}{\ref{th:main_real_vertex}}
\begin{theorem}
Let $G$ be a planar graph with positive edge capacities. Then for any $c,\delta>0$, we compute a value $vit^\delta(v)\in(vit(v)-\delta,vit(v)]$ for all $v\in V(G)$ satisfying $c(v)\leq c$, in $O(\frac{c}{\delta}n+n\log n)$  time.
\end{theorem}
\addtocounter{theorem}{-1}
}
\begin{proof}
We compute $D$ and $U$ in $O(n\log\log n)$ time by Theorem~\ref{th:U_italiano+err_giappo}. If $c(v)<c$, then $w(f^D_v)<c$. For convenience, we denote $f^D_v$ by $f$. By Lemma~\ref{lemma:centro_faccia_nlogn}, we compute $\dist_{D}(f,q^y_f)$ (resp., $\dist_{D}(f,q^x_f)$)  in $O(n\log n)$ time, for all $f\in F^y$ (resp., for all $f\in F^x$). 
 Now we have to show  how to obtain $\min_{i\in[k]}\{d_i(f)\}$, $\min_{i\in[k]}\{d_i(q^y_f)\}$ and $\min_{i\in[k]}\{d_i(q^x_f)\}$ that we may compute with an error depending on $\delta$.

We note that $\min_{i\in[k]}d_i(f)=\min_{i\in[k],d_i<\MF+w(f)}d_i(f)$. Indeed, if $d_i(f)=\MF-z$, for some $z>0$, then $d_i$ is at most $\MF-z+w(f)$. For the same reason, $\min_{i\in[k]}d_i(q^x_f)=\min_{i\in[k],d_i<\MF+w(f)}d_i(q^x_f)$ and, similarly, $\min_{i\in[k]}d_i(q^y_f)=\min_{i\in[k],d_i<\MF+w(f)}d_i(q^y_f)$.

By using Lemma~\ref{lemma:cost_L_r}, for each $r\in\{0,1,\ldots,\lceil\frac{c}{\delta}\rceil\}$, we compute a value $\alpha_r(f)\in[\min_{i\in L_r}d_i(f),$ $\min_{i\in L_r}d_i(f)+\delta)$, for all $f\in  F$, in $O(n)$ time. Then, for each $f\in  F$, we compute $\alpha(f)=\min_{r\in\{0,1,\ldots,\frac{c}{\delta}\}}\alpha_r(f)$. By above, for any $f\in F$ satisfying $w(f)<c$, it holds that $\alpha(f)$ satisfies $\alpha(f)\in[\min_{i\in[k]}\{d_i(f)\},$ $\min_{i\in[k]}\{d_i(f)\}+\delta)$. 



With a similar strategy, by replacing Lemma~\ref{lemma:cost_L_r} with Lemma~\ref{lemma:cost_L_r_vertex}, for each $f\in F^y$ satisfying $w(f)<c$ and $\min_{i\in L_r}d_i(q^y_f)<\dist_{D}(f,q^y_f)$, we compute a value $\beta(f)\in[\min_{i\in[k]}\{d_i(q^y_f)\},$ $\min_{i\in[k]}\{d_i(q^y_f)\}+\delta)$. The same results hold for the ``$x$ case'': for each $f\in F^x$ satisfying $w(f)<c$ and $\min_{i\in L_r}d_i(q^x_f)<\dist_{D}(f,q^y_f)$, we compute a value $\gamma(f)\in[\min_{i\in[k]}\{d_i(q^x_f)\},$ $\min_{i\in[k]}\{d_i(q^x_f)\}+\delta)$.

Then, by Proposition~\ref{prop:vitality_vertices}, for each $v\in V(G)$ satisfying $c(v)\leq c$ ($w(f)<c$) we compute a value $vit^\delta(v)$ satisfying $vit^\delta(v)\in(vit(v)-\delta,vit(v)]$ in $O(1)$ time by using $\dist_{D}(f,q^x_f)$, $\dist_{D}(f,q^y_f)$, $\alpha(f)$, $\beta(f)$ and $\gamma(f)$.\qed
\end{proof}

{
\renewcommand{\thetheorem}{\ref{th:brutto}}
\begin{theorem}
Let $G$ be a planar graph with positive edge capacities. Then for any $S\subseteq V(G)$, we compute $vit(v)$ for all $v\in S$ in $O(|S|n+n\log\log n)$  time.
\end{theorem}
\addtocounter{theorem}{-1}
}
\begin{proof} 
We compute $D$ and $U$ in $O(n\log\log n)$ time by Theorem~\ref{th:U_italiano+err_giappo}.  For convenience, we denote $f^D_v$ by $f$. To compute $\min_{i\in[k]}\{d_i(f)\}$, $\dist_{D}(f,q^x_f)$ and $\dist_{D}(f,q^y_f)$ we put a vertex $u_f$ in the face $f$ and we connect it to all vertices of $f$ with zero weighted edges. Then we solve an SSSP instance with source $u_f$ and we compute $\min_{i\in[k]}\{d_i(f)\}$, $\dist_{D}(f,q^x_{f})$ and $\dist_{D}(f,q^y_f)$ in $O(n)$. With a similar strategy, we compute $\min_{i\in[k]}\{d_i(q^y_f)\}$ and $\min_{i\in[k]}\{d_i(q^x_f)\}$ in $O(n)$ time. Finally, by Proposition~\ref{prop:vitality_vertices}, for each $v\in S$, we compute $vit(v)$ in $O(1)$ time.\qed
\end{proof}

\section{Small integer capacities and unit capacities}\label{sec:small_integer}

If the edges capacities are integer, then we can compute the max flow in $O(n+L)$ time~\cite{eisenstat-klein} and also $U$ in $O(n+L)$ time~\cite{balzotti-franciosa_2}, where $L$ is the sum of all the edges capacities. 
{
\renewcommand{\thecorollary}{\ref{cor:integer}}
\begin{corollary}
Let $G$ be a planar graph with integer edge capacity and let $L$ be the sum of all the edges capacities. Then
\begin{itemize}\itemsep0em
\item for any $H\subseteq E(G)\cup V(G)$, we compute $vit(x)$ for all $x\in H$, in $O(|H|n+L)$ time,
\item for any $c\in\mathbb{N}$, we compute $vit(e)$ for all $e\in E(G)$ satisfying $c(e)\leq c$, in $O(cn+L)$ time.
\end{itemize}
\end{corollary}
\addtocounter{corollary}{-1}
}
\begin{proof}
Note that, being all the edge capacities integer, then every edge or vertex vitality is an integer. Thus, by taking $\delta=1$ in Theorem~\ref{th:main_real_edge} and Theorem~\ref{th:main_real_vertex}, we obtain all the vitalities without error. The two statements follow from the proof of Theorem~\ref{th:main_real_edge} and Theorem~\ref{th:main_real_vertex} by taking $\delta=1$ and by computing $U$ in $O(n+L)$ time instead of $O(n\log\log n)$ time by using algorithm in~\cite{balzotti-franciosa_2}.\qed
\end{proof}

{
\renewcommand{\thecorollary}{\ref{cor:unweighted}}
\begin{corollary}
Let $G$ be a planar graph with unit edge capacity. Let $n_{>d}$ be the number of vertices whose degree is greater than $d$. We can compute the vitality of all edges in $O(n)$ time and the vitality of all vertices in $O(\min\{n^{3/2},n(n_{>d}+d+\log n)\})$ time.
\end{corollary}
\addtocounter{corollary}{-1}
}
\begin{proof}
The complexity of edge vitality is implied by Corollary~\ref{cor:integer} by taking $c=1$ and because $L=O(n)$. Being the vitality integers, then we compute the vitality of all vertices in $O((n_{>d}+d)n+n\log n)$ time by Theorem~\ref{th:main_real_vertex}.

To compute the vitality of all vertices in $O(n^{3/2})$  time we note that in a planar graph, by Euler formula, there are at most $6\sqrt{n}$ vertices whose degree is greater than $\sqrt{n}$. Thus it suffices to take $d=\sqrt{n}$, that implies $n_{>d}\leq 6\sqrt{n}$ and $O((n_{>d}+d)n+n\log n)=O(n^{3/2})$.\qed
\end{proof}

Kowalik and Kurowski~\cite{kowalik-kurowski} described an algorithm that, given an unweighted planar graph $G$ and a constant $d$, with a $O(n)$ time preprocessing can establish in $O(1)$ time if the distance between two vertices in $G$ is at most $d$ and, if so, computes it in $O(1)$ time.
{
\renewcommand{\thecorollary}{\ref{cor:bounded}}
\begin{corollary}
Let $G$ be a planar graph with unit edge capacity where only a constant number of vertices have degree greater than a fixed constant $d$. Then we can compute the  vitality of all vertices in $O(n)$ time.
\end{corollary}
\addtocounter{corollary}{-1}
}
\begin{proof}
By above discussion, Corollary~\ref{cor:unweighted} and the proof of Theorem~\ref{th:main_real_vertex}, it suffices to show that we can compute $\dist_{D}(f_v,q^y_f)$ in $O(n)$ worst-time for all $f\in F^y$.
For every $i\in[k]$ let $F_i\subseteq F^y$ be the set of faces such that $x_i\in f$, for all $f\in F_i$. Note that if $|F_i|=m$, then $deg_{D}(x_i)=m+1$.

Let $d$ be the maximum degree of $G$. 
We need $d_{{D}}(y_i,z)$, for all $i\in[k]$ and $z\in V(f)$, for all $f\in F_i$. If we use the algorithm in~\cite{kowalik-kurowski}, then we spend $\sum_{i\in[k]}\left(\sum_{f\in F_i}|V(f)|\right)\leq\sum_{i\in[k]}|F_i|d =\sum_{i\in[k]}(deg_{{D}}(x_i)+1)d=O(n)$ total time.\qed
\end{proof}

\section{Conclusions and open problems}\label{sec:conclusions}

We proposed algorithms for computing an additive guaranteed approximation of the vitality of all edges or vertices with bounded capacity with respect to the max flow from $s$ to $t$ in undirected planar graphs. These results are relevant for determining the vulnerability of real world networks, under various capacity distributions.

It is still open the problem of computing the exact vitality of all edges of an undirected planar graph within the same time bound as computing the max flow value, as is already known for the $st$-planar case.


\end{document}